\documentclass[11pt]{article} 

\usepackage{amsmath,amsfonts,amsthm,amssymb} %
\usepackage[pdftex]{graphicx}
\usepackage[letterpaper, margin=1in]{geometry} %
\usepackage{setspace} 
\usepackage{mathrsfs} %
\usepackage{color}
\usepackage{float} %
\usepackage{algorithm}
\usepackage{algpseudocode}
\usepackage{amsthm}
\usepackage{subfig}
\usepackage{comment}

 \usepackage{multirow}
 \usepackage[table,xcdraw]{xcolor}
 \usepackage{hyperref,xcolor}
 
 \usepackage[backgroundcolor=white, bordercolor=blue, linecolor=blue, textsize=tiny]{todonotes}

\hypersetup{linkcolor=blue}

\usepackage{bm}
\usepackage{url}
\usepackage[shortlabels]{enumitem}

\graphicspath{{Images/}}
\DeclareGraphicsExtensions{.pdf,.jpeg,.png}

\usepackage[round,authoryear,numbers,sort]{natbib}

\newtheoremstyle{examplestyle}%
   {\topsep}%
   {1pt}%
   {\slshape}%
   {0pt}%
   {\bfseries}%
   {}%
   {1mm}%
   {}%
\theoremstyle{examplestyle}

\newtheorem{thm}{Theorem}[]
\newtheorem{prop}{Proposition}[]
\newtheorem{lemma}{Lemma}[]
\newtheorem{defi}{Definition}[]

\numberwithin{equation}{section} %

\usepackage{multirow}
\usepackage{authblk}

\newcommand{\vpNew}[1]{\textcolor{black}{#1}}

\begin{document}

\title{Time-Equitable Dynamic Tolling Scheme For Single Bottlenecks}

\author[1] {John W Helsel}
\author[2]{Venktesh Pandey \thanks{ Corresponding author: venktesh\textrm{@}utexas.edu}}
\author[2]{Stephen D Boyles}

\affil[1]{\small \vpNew{WSP USA, Sacramento, California, USA}}
\affil[2]{\small Department of Civil, Architectural and Environmental Engineering, The University of Texas at Austin, Austin, Texas, USA}

\markboth{Helsel et al. 2019}%
{Helsel \MakeLowercase{\textit{et al.}}: Time-equitable toll}

\date{ }
\maketitle

\begin{abstract}

Dynamic tolls present an opportunity for municipalities to eliminate congestion and fund infrastructure. \vpNew{Imposing tolls} that regulate travel along a public highway through monetary fees raise worries of inequity. In this article, we introduce the concept of time poverty, emphasize its value in policy-making in the same ways income poverty is already considered, and argue the potential equity concern posed by time-varying tolls that produce time poverty. We also compare the cost burdens of a no-toll, system optimal toll, and a proposed ``time-equitable" toll on heterogeneous traveler groups using an analytical Vickrey bottleneck model where travelers make departure time decisions to arrive at their destination at a fixed time. We show that the time-equitable toll is able to eliminate congestion while creating equitable travel patterns amongst traveler groups.

\vspace{3mm}
\noindent{ \textit{\textbf{Keywords}}: \vpNew{Equity}, Time poverty, Vickrey bottleneck model, Congestion pricing }
\end{abstract}

\newpage
\section{Introduction}
\label{sec:intro}
	Congestion pricing is a commonly used mechanism for relieving congestion on a roadway along with generating revenue for infrastructure projects. Peak-period congestion pricing, commonly implemented using dynamically priced toll lanes and cordon tolls, charges higher prices for travelers in the middle of the peak period.  \vpNew{The intent is to force travelers to}  more explicitly consider the cost of travel and may choose to shift the time of travel for trips that are less important. 
	
	Contemporary research on the dynamic pricing of toll roads has largely been concentrated on their ability to effectively relieve congestion, improve environmental pollution measures, and maximize profits for the road operator. For example, cordon tolls in Sweden have reduced congestion by 22\% and enjoy widespread public support~\citep{sorensen2014strategies}. Similarly, a survey of European congestion tolls saw 14\% and 23\% decreases in vehicle counts in Milan and Bologna; 30\% and 33\% decreases in delay times in London and Stockholm; 13-21\% reduction in CO$_2$ emissions in Stockholm and Rome; and a 14\% decline in collision rates in Milan~\citep{may2009curacao}. These effects are predictable and reproducible across a wide variety of contexts, because consumer elasticity towards tolls has been found to stably vary between $-0.2$ and $-0.8$ across a wide variety of research~\citep{wuestefeld1981impact, white1984man, ribas1988estudi, jones1992restraining, hirschman1995bridge, gifford1996demand, odeck2008travel}.

	\vpNew{Although dynamic tolls have succeeded in reducing congestion, they raise other worries including inequity.} Real world practitioners are keenly aware of the equity pitfalls congestion tolls present. \vpNew{One position is} that equity concerns are a purely public relations~\citep{wang2012tradable} or an education~\citep{giuliano1992assessment} problem. Commonly recommended solutions to increase support for toll roads and assuage worries over equity include allocating toll revenues to projects that visibly increase equity, using the revenues for subsidizing public transit, or evenly distributing revenues amongst travelers as a substitution for gasoline or other sales taxes~\citep{small1992using}.

	While some equity concerns can be addressed through public outreach, in this article we \vpNew{argue that inequity should be addressed more directly} in the design of dynamic tolls. Commonly used dynamic tolls are designed for system optimal efficiency. Such tolls minimize the total cost incurred \vpNew{in} the system by charging a dynamic toll that applies uniformly to all travelers. A uniform toll predicates access to roads at peak periods on one's ability to pay and decreases the likelihood that poorer travelers will choose to travel at peak periods. However, if nondiscriminatory access to public roads is a public good, then this observation is a \textit{prima facie}, though defeasible, case against system optimal tolls. In addition, research on time poverty by sociologists raises worries that there are significant harms to low-income travelers associated with being segregated to off peak travel over and above any concerns raised about the monetary costs of the tolls.
	
	Consider a small motivating example. Suppose there are two travelers, X and Y, who want to arrive at the same destination at the same time and must pass through the same bottleneck. Traveler X receives a penalty of $10$ money units for arriving late, while traveler Y receives a penalty of $9$ money units for the same. The operator at the bottleneck has to decide who goes through the bottleneck first; the traveler going first arrives on time and incurs no penalty, while the other traveler incurs a penalty. A system optimal solution considers travelers' absolute costs and selects an order than minimizes the total system cost: traveler X goes first, while traveler Y waits and pays a penalty of $9$ units. However, traveler X is rich and owns $100$ units of money, while traveler 2 is poor and owns $30$ units of money. The relative penalty for arriving late is 10\% of the current worth for traveler X, while it is 30\% for traveler Y. Under the system optimal solution traveler Y ends up significantly worse off as they end up losing 30\% of their worth. 
	
	This example, \vpNew{though not a precise analog to dynamic tolling}, carries the key essence of our argument: equity issues arise when a traveler's \vpNew{\textit{relative}} worth of time is ignored. A system optimal toll, that values travelers based on their absolute willingness to pay, may minimize \vpNew{cost} for the whole system, but it can also increase the equity gap across the population leaving the poor travelers worse off.  Considerations of equity show that the assumption about tolling based on absolute willingness to pay should be challenged. In particular, research into time poverty shows that poorer travelers already have to substitute their time for goods/service out of their price range. This article brings together a discussion on toll schemes, equity, and time poverty within transportation policy. Although the literature has discussed all three of these issues and even brought some of them together~\citep[see][]{arnott1987schedule,lindsey2004existence, litman1996using,akamatsu2017tradable}, there has been no sustained discussion of how transportation policies can exacerbate or ameliorate inequitable distributions of time poverty.

	In this article, we analyze the issue of time poverty for dynamic tolls imposed over a single bottleneck that a group of travelers use in order to arrive at their destination at a given time. First, we show that equity concerns in transportation are primarily rooted in a desire to respect all travelers equally and that time poverty ought to be considered in policy-making in the same ways that income poverty already is. We argue that tolling schemes that produce time-poverty among poorer travelers ought to be examined as a potential equity concern. Second, we extend the qualitative equity investigation to a particular implementation of a time-varying toll on single bottlenecks to examine whether it raises equity concerns. We also propose an alternative, which we call a \textit{time-equitable toll}, that eliminates congestion while creating equitable travel patterns amongst traveler groups \vpNew{by charging different rates for different groups}.

	The primary contributions of this article are two-fold: (a) we conduct an equity analysis of the system optimal dynamic toll highlighting the importance of equity for time as a resource, and (b) we develop a time-equitable tolling scheme and \vpNew{show how} it addresses these inequities.  %
	The rest of the article is organized as follows. Section~\ref{sec:equityPhilosophy} presents a review of equity in the context of transportation systems and introduces the idea of time poverty. Section~\ref{sec:vickreyPrelim} discusses some of the preliminary ideas from Vickrey's bottleneck model~\citep{vickrey1969congestion}. Section~\ref{sec:timeEquitableToll} defines the time-equitable tolls and proves its efficiency relative to a system-optimal toll. Section~\ref{sec:numerical} presents numerical results and Section~\ref{sec:conclusion} concludes the article and provides directions for future work.

\section{Equity and Time Poverty: Literature Review}
\label{sec:equityPhilosophy}
In this section, we present how equity is defined and discussed for toll roads in the literature. We also show the need for including time poverty in the discussions surrounding equity, which in the context of transportation systems means arriving to work on time and not being forced to leave early to avoid tolls.

\subsection{Equity concerns for toll roads}
Given the complexities of real world policies, their unpredictable interactions, and the way support for policies in one arena motivate our preferences in others, judgments about the equity impacts of a particular policy are often unclear and uncertain. %
\cite{litman2012evaluating} provides a simpler definition as follows:

\begin{defi}[Equitable policies~\citep{litman2012evaluating}]
An equitable policy is the one that fairly allocates its costs and benefits.
\end{defi}

This definition requires estimating costs and benefits across different social \vpNew{groups}, which is a difficult task. There are three diverse conceptions of equity in the transportation sector based on how the perceived cost and benefits are allocated~\citep{litman2012evaluating}.
\begin{enumerate}
	\item \textbf{Horizontal equity}: A policy exhibits horizontal equity when its costs and benefits are distributed strictly by cost or need. The goal is perfectly equal treatment; individuals subject to the policy should receive equal benefits, pay equal costs, and be treated in a procedurally identical manner. Example policies include the allocation of \vpNew{one vote to each person}, and movements to ensure that local property taxes \vpNew{are} spent on local projects.
	\item  \textbf{Vertical equity with regard to demography or income}: Vertically equitable policies distribute their costs and benefits with a sensitivity to the distribution of impacts between groups that differ with respect to income, social class, race, or some other identifiable distinction. Transport policies may be vertically equitable if they favor disadvantaged groups in a manner that compensates for structural inequalities in the larger society~\citep{rawls2009theory}. %
	For example, vertically equitable transit policies may call for higher levels of service to poorer neighborhoods in recognition of the fact that many residents lack alternative transport modes.
	\item \textbf{Vertical equity with regard to mobility need and ability}: A second variety of vertically equitable policies focuses on the different ways differing mobility needs and abilities change access to social goods. Policies under this aspect of equity include universal design such as curb cuts for wheelchairs, audible beacons for crosswalks, or wheelchair ramps on kneeling buses.
\end{enumerate}

	This \vpNew{raises the question of which conceptions of equity are best suited to address dynamic tolls}? Before \vpNew{proposing an} answer, we review  equity concerns for toll roads in the literature. 

\paragraph{Equity arguments against toll roads} Tolls roads have long been accused of increasing inequity in a number of ways. Toll roads charge a financial cost for accessing a facility and impose burden on travelers in three ways. First, and \vpNew{most obviously}, toll roads impose financial burden on poor travelers. \cite{harvey1991suitability} and \cite{giuliano1992assessment} note that tolls create winners and losers and do not lead to a Pareto improvement as workers in areas with poor transit are highly likely to pay bridge tolls at peak hours because they lack alternatives. Second, toll roads create environmental costs and congestion in poorer neighborhoods when travelers \vpNew{choose} to use local roads in order to avoid the toll. For example, \vpNew{New York City} Mayor Michael Bloomberg's 2007 plan to institute a congestion charge in lower Manhattan was opposed, as it was shown that the funds raised would largely come from lower-middle class families~\citep{taylor2010addressing}. Last, toll roads limit access to transportation facilities for poor travelers. Often dynamic tolling systems work best with cashless transactions, but poorer travelers are less likely than affluent travelers to have credit cards or bank accounts, which makes it both more difficult and more expensive to access facilities with electronic tolling systems~\citep{plotnick2009impacts}. Similarly, many poor travelers face constrained schedules that effectively force them to travel at peak periods. For instance, single mothers must pick up and drop their kids off in narrow windows, and \vpNew{therefore} avoid nonstandard working hours if at all possible~\citep{presser1995job}.

\paragraph{Equity arguments in favor of toll roads} \vpNew{Despite the concerns about the inequitable impact of toll roads, some researchers} have argued that tolls do not have a uniform impact on equity because their effect on the road network is highly sensitive to the specifics of the proposed plan, particularly its price structure, quality of alternatives, use of raised revenues, and whether driving is a luxury or a necessity, i.e. whether transit or routes that avoid the toll are available~\citep[see][]{litman1996using, litman2012evaluating, raje2003impacts, golub2010welfare, schweitzer2009empirical, santos2004distributional}. These authors concede that particular implementations of time-varying tolls may have regressive income effects, \vpNew{but argue that} tolls can also be fine-tuned to have progressive outcomes. For example, express lane facilities provide a non-tolled alternative and are claimed to be more equitable than standard toll roads by providing a congestion-free alternative for urgent trips for a moderate fees~\citep{fhwaEquity}.

Additionally, other researchers have argued that the equity implications of a time-varying toll are entirely dependent on the alternative methods available for funding transportation projects. Unfortunately, the commonly used alternative --- a sales tax --- is regressively inequitable in itself due to three reasons. First, user fees are regressive \vpNew{by nature}; poorer travelers have less money and so \$1 in tax to them represents more of their income than a \$1 tax would represent for wealthier travelers. Second, the poor are also the most likely to own older and less fuel-efficient vehicles. They spend more money (which is a larger percentage of their income) to go fewer miles, \vpNew{a disparity that} holds for all travel, not just for the most congested travel. Last, poorer travelers tend to live in the urban core and drive less frequently, so area-wide sales taxes tend to force them to subsidize infrastructure that they do not use. Thus, non-discriminating sales taxes pose problems for horizontal equity among travelers. 
	
	These arguments suggest that financing transportation projects using tolls as a  revenue stream is a viable option, and provided the toll implementation addresses the social inequalities, it can be an equitable alternative to using sales tax. %
	In fact, there are mechanisms to implement tolls that benefit the poor. A Florida study examining strategies to decrease vehicle miles traveled found that a tradeable credit scheme would benefit low-income households because they would be able to trade away their credits and so receive compensation for traveling less~\citep{mamun2016comparison}. A version of the tradeable credit scheme was also implemented in Minnesota with some success. The state toll road authority, MnPass, credited every account with a positive balance that could be used and supplemented with a user's funds. If a user chose not to travel on toll lanes, however, the funds could be applied towards registration fees, a cost savings that would benefit all drivers~\citep{taylor2010addressing}.

	In this article, we are most concerned with whether such a time-varying toll ought to be implemented in a horizontally equitable manner (that is, the only determinant of access is one's willingness and ability to pay the toll) or whether there are identifiable socioeconomic groups that, for a reason we will need to articulate, ought to receive a subsidy on their travel relative to the general public. While there are many interesting research questions into how we might devise a tolling scheme that addresses Litman’s third kind of equity or even racial inequities, we only focus on parameters of the Vickrey bottleneck that are most clearly applied to income and socio-economic status. %

\subsection{Justifying vertically-equitable tolls}

The toll policies proposed in the literature commonly assume, implicitly or explicitly, a particular theory of value, viz. ``system optimal efficiency" where the goal is to minimize \vpNew{total cost in} the system. This theory of value says that willingness to pay should be used as a way of sorting who has the right to travel at a given time. However, minimizing costs is not guaranteed to be most equitable in terms of fairly distributing the costs and benefits. (See the example in Section~\ref{sec:intro}.)  Another commonly-used objective for toll policies is revenue maximization; however, even revenue-maximizing tolls can have unintended consequences, creating differential impact on different groups. For example, maximizing revenue on express lane facilities promotes ``jam-and-harvest" \vpNew{strategies} where it is in the toll operator's best interest to increase congestion on the regular lanes in the early peak period to harvest more revenue towards the latter period~\citep{goccmen2015revenue, pandey2018dynamic}. Additionally, jam-and-harvest increases costs on poor travelers in favor of reducing costs for wealthier travelers. 

When we consider a potential policy, we can ask of it: ``Does this policy show equal respect to all citizens?" In cases where horizontal equity is the primary social good, we judge that the citizens are on an equal playing field, therefore we show them equal respect by treating everyone in the same manner. In cases where vertical equity is important, we judge that citizens are not on an equal playing field and so we show them equal respect by ensuring that the least advantaged receive compensatory benefits that allow them to access or compete for social and economic goods from which they would otherwise be shut out.\footnote{This article is an inappropriate forum to describe all of the debates that have taken place in the philosophical literature regarding equity and justice. We refer the reader to the work by~\cite{rawls2009theory} and \vpNew{Chapter} 2 of~\cite{helsel2017getting}.} 

Developing equitable tolling policies is challenging due to the difficulty in defining when \vpNew{citizens are} on equal playing fields. However, in cases where we \vpNew{do observe} systematic inequalities, such as the differential treatment of groups due to the system-optimal or revenue-maximizing tolls, a vertically-equitable toll is justified and ought to be considered. %

\subsection{\vpNew{Introducing time poverty analysis}}
\label{subsec:timePoverty}
Time poverty should be included in equity analysis because it has the same importance to travelers as unfairly distributed economic costs. First, low-income travelers experience poverty both with respect to money but also with respect to their time. Second, time poverty creates concerns that low-income travelers will be excluded from primary social goods. Third, the socially optimal dynamic tolling plan exacerbates these problems. Therefore, the dynamic tolling plan creates equity concerns and we should examine alternatives that mitigate this issue.

The concepts of time-poverty are not new. \cite{vickery1977time} argued that measurements of poverty based solely upon household income miss an important aspect of poverty, namely that the time to fulfill basic household needs increases as income decreases. One definition of time poverty is presented in~\cite{bardasi2006measuring}.

\begin{defi}~\cite[Time poverty,][]{bardasi2006measuring}
``Conceptually, time poverty can be understood as the fact that some individuals do not have enough time for rest and leisure after taking into account the time spent working, whether in the labor market, for domestic work, or for other activities such as commuting."
\end{defi}

Families that do not fall into the income definition of poverty can still experience time poverty. This is because even when extreme want is out of the equation, time poverty is a problem driven by the number of hours worked, when they are worked, and the physical intensity of that work~\citep{warren2003classand}. %
Higher paid workers can control when they work, where they work, and have more flexibility to take time off, which creates greater job satisfaction, lowers stress levels, and creates work-life balance~\citep{doyle2001time}. Moving up the socio-economic scale also allows people to ``buy back" their time by paying others to do domestic chores~\citep{gregson2005servicing, roberts1998work, stephens1999fight}.

 Although the literature covers a wide variety of perspectives, the discussion of time poverty has often been framed as a problem facing female managers as they struggle to balance long and demanding workdays with the traditional responsibilities of a homemaker~\citep{rutherford2001organizational}. But this focus on the upper-middle class and corporate executives misses the crucial experience of millions of low-income families. As families move out of monetary poverty, they can find themselves suddenly time poor. For instance, 50--60\% of unemployed single parents are below the monetary poverty threshold, but only 3\% are time poor~\citep{harvey2007twenty}. When single parents are employed only 26--31\% are monetarily poor, but 98\% are time poor. Only 5.3\% of single parents are neither time nor monetarily poor. Even a partner to share the load is not a panacea, since 20--30\% of employed two-parent families are still time poor (they work above the threshold of 11.5 hours per day per parent).

We argue that time poverty ought to be a central focus of our equity evaluation because time is one of the most precious goods humans possess. Intrinsically, free time provides an opportunity for rest, social interaction, leisure participation, and self-realization, which makes it an important non-monetary welfare resource~\citep{chatzitheochari2012class}. Philosophers, economists, and social theorists have consistently conceptualized free time as a primary good for individual well-being~\citep{marx1968karl,blackden2006gender}. %
This right is so fundamental that it is recognized in a number of international treaties, including the Universal Declaration of Human Rights~\citep{UN2020} and the International Covenant on Economic, Social and Cultural Rights~\citep{unhrICESC}. These political documents recognize the fact that human beings work in order to participate in those activities and to be with those people who make life interesting and fun.

The \vpNew{principle} that equitable policies show equal respect for all citizens~\citep{rawls2009theory} helps us to understand why time as good should be shared by all citizens and not just the wealthy. In context of time-varying tolls this is a key concern where the commonly used system optimal objective can have unintended consequences of penalizing relatively time poor travelers. In the next section, we highlight the case where dynamic tolling for single bottlenecks poses concerns of time equity and propose a vertically-equitable toll to address this inequity.

\section{Preliminaries: Vickrey's Bottleneck Model}
\label{sec:vickreyPrelim}
In this section we introduce the multiclass Vickrey’s bottleneck model~\citep{vickrey1969congestion}. The simplest version of this model includes a single link between the origin and the destination.  The link has an unlimited capacity before and after the bottleneck, but at the bottleneck the capacity is limited. (This capacity assumption mimics a point queue model in dynamic traffic assignment.) Vickrey's model is useful as it offers insights into ``the nature of time-of-day shifts, various inefficiencies in unpriced equilibria, the temporal pattern of 
optimal pricing, and some surprising effects of pricing on travel patterns and travel costs."~\citep{small2015bottleneck}. %
Positioning this article among the variants of Vickrey's model, we label the specific assumptions as A$\#$. Through the rest of the article, we assume vehicles have single occupancy, and use travelers and vehicles interchangeably (A$\#1$).

\subsection{Notations and assumptions}
Trips along the Vickrey's bottleneck model have five stages: a departure, a free flow period to the bottleneck, a congested phase while waiting for prior vehicles to pass the bottleneck, passing the bottleneck, and an arrival. \vpNew{The simplest} version of the model takes three inputs: number of travelers~($N$), a fixed rate of discharge at the bottleneck in vehicles per time unit ($D$), and \vpNew{a desired} arrival time at the destination for all vehicles ($\tau^*$). We assume inelastic demand, fixed capacity, and \vpNew{a fixed desired} time of arrival for all vehicles to simplify our analyses (A$\#2$). \vpNew{This assumption allows us to focus our analysis on equity as clearly as possible.}

All travelers possess a value of time (expressed as money units per time unit), which represents the perceived cost of a unit of time spent traveling or waiting in a queue. Travelers also have schedule delay penalties for early and late arrival. In the commuting case, these penalties are influenced by whether the worker has a flexible schedule that permits them to be productive before the desired start time, the magnitude of penalties the worker faces for arriving late, and other schedule constraints that may make it costly for the worker to arrive before or after their start time (for example, convenience charges by childcare services to drop a child off early).

\vpNew{Travelers are assumed to belong to different heterogeneous groups with varying values of time and schedule delay penalties} (A$\#3$). Modeling continuous distributions of parameters requires solutions to nonlinear differential equations~\citep{arnott1990economics} and is left as part of the future work. Let $K$ denote the set of all groups, and $\alpha_k, \beta_k$, and $\gamma_k$ (respectively) denote the value of time, the early arrival penalty, and the late arrival penalty, for a traveler in group $k \in K$.

Let $\tau(t)$ denote the destination arrival time of a traveler departing from the origin at time $t$. \vpNew{Because each group contains more than one traveler, let $n_k(t)$ denote the departure rate of group $k$ at time $t$} and $N_k$ denote the total number of travelers in group $k$. Further, define $n(t)=\sum_{k\in K}n_k(t)$ as the total number of travelers departing at time $t$. While the physical interpretation of such a model requires that vehicles be discrete entities, we assume infinitely divisible travelers over a continuous time to \vpNew{facilitate analysis} (A$\#4$).

The utility of each traveler in group $k$ departing at time $t$ is captured in a cost function, $c_{\text{total}}^k(t)$, which is composed of four distinct costs: the cost (determined by $\alpha_k$) of the time required to travel from the origin to the destination in free flow conditions, $c_{\text{FF}}^k(t)$; the cost of the time required to wait in any queue at the bottleneck, $c_{\text{queue}}^k(t)$; the schedule delay cost (determined by $\beta_k$ and $\gamma_k$) created by arriving at the destination at a time other than the desired arrival time, $c_{\text{SD}}^k(t)$; and, the direct cost of a toll that is imposed, $c_{\text{toll}}^k(t)$. Each of these costs can vary over the peak hour and it is possible for \vpNew{some of these} costs to be equal to 0. 

Without loss of generality, we set $c_{\text{FF}}^k(t)=0$ for all $t$ and $k\in K$. A constant free flow term merely adds a constant term to all of the cases in \vpNew{our analysis, which does not affect comparisons across alternative toll schemes}. Let $SDC_k$, $TTC_k$, $TRC_k$, and $TC_k$ be the schedule delay cost, travel time cost, toll rate cost, and total cost across all travelers in group $k$. Without the subscript, let each of these terms represents the total of the cost across all groups. That is, $SDC = \sum_{k\in K} SDC_k$, $TTC = \sum_{k\in K} TTC_k$, $TRC = \sum_{k\in K} TRC_k$, and $TC = \sum_{k\in K} TC_k$. Later in the text, we will use a superscript on the cost terms to differentiate costs for different toll scenarios.

Several \vpNew{basic relations hold among the defined variables}. First, Equation~\eqref{eq:tauAndtrelation} evaluates arrival time $\tau(t)$ by adding the time spent in the queue to the departure time: 
	\begin{equation}
		\tau(t) = t + \frac{q(t)}{D},
		\label{eq:tauAndtrelation}
	\end{equation}
	
	\vpNew{where $q(t)$ is the queue length at time $t$.} Second, flow conservation requires that the rate of change of the number of vehicles in a queue is equal to the difference between the inflow rate and the outflow rate.  Mathematically,
	\begin{equation}
		\frac{dq(t)}{dt} = n(t)-D.
		\label{eq:qLengthDer}
	\end{equation}
	
	 \vpNew{Equation~\eqref{eq:qLengthDer} holds only when the bottleneck is active, that is $n(t)\geq D$.} Next, the expressions for cost terms can be evaluated. \vpNew{We assume that $c_{\text{queue}}^k(t)$ is derived from a model of a deterministic queue (A$\#5$), that is, the queue delay is evaluated as:}
	\begin{equation}
		c_{\text{queue}}^k(t) = \alpha_k. \frac{q(t)}{D}.		
	\end{equation}
	The schedule delay cost \vpNew{distinguishes} the two cases of early or late arrival and is expressed a maximum between the two terms in Equation \eqref{eq:SDcost}. As an extension of assumption A$\#2$, $\tau^*$ is same across all groups.
	\begin{equation}
		c_{\text{SD}}^k(t) = \max \left\lbrace \beta_k \left( \tau^*-\tau(t) \right), \gamma_k \left(\tau(t)-\tau^*\right) \right\rbrace.
		\label{eq:SDcost}
	\end{equation}

	We model travelers as purely rational and selfish, that is, each traveler wants to minimize her own total cost (A$\#6$). Thus, the choice of departure time for each traveler can be modeled as a non-cooperative game among non-atomic travelers and the steady-state behavior of the game is considered. \cite{lindsey2004existence} defines the deterministic departure-time user equilibrium as \vpNew{the} following:

\begin{defi}[Deterministic departure-time user equilibrium]
 At equilibrium, all travelers in a group incur the same trip cost for their chosen departure times, and equal or higher costs at any other times.
 \label{def:equil}
\end{defi}

Because the arrival times are uniquely determined by the departure time, the trip cost for all travelers in a group are also same as a function of arrival times. If we assume that $\alpha_k > \beta_k$ for each group (A$\#7$)\footnote{This is a reasonable assumption because if $\alpha_k < \beta_k$, then we would be modeling travelers who derive more benefit from waiting in traffic than arriving early to do paperwork or even sitting in the company parking lot. The empirical findings in~\cite{small1982scheduling} also support this assumption.}, and given that the linear schedule delay cost functions in Equation \eqref{eq:SDcost} are continuous, finite, and evaluate to zero for $\tau(t)=\tau^*$ for each group $k$, the equilibrium departure time distribution for each group exists and is unique~\citep[][Propositions 1 and 2]{lindsey2004existence}. \vpNew{The equilibrium departure time distributions can be found by solving a mixed linear complementarity problem  using algorithms like Lemke's}~\citep{ramadurai2010linear}. 

In this article, we make additional assumptions that help us derive analytical formulas for the equilibrium, obviating the need for mixed linear complementarity formulations. First, we limit our attention to two groups. This can be done without loss of generality as the formulas in the following sections can be easily extended for three or more travel groups. Second, \vpNew{as others have assumed} in the literature~\citep{henderson1981economics, arnott1987schedule, arnott1990economics}, we assume that the ratio of late and early arrival penalties are constant for both groups (A$\#8$). That is, $\gamma_1/\beta_1 = \gamma_2/\beta_2 = \eta$, \vpNew{where $\eta$ is a fixed constant greater than 1}. 

Next, we classify the groups based on the absolute and relative values of $\alpha$ and $\beta$. The value of $\alpha$ is typically proportional to the income level of a class, thus we call travelers with higher $\alpha$ as \textit{high-income} travelers. A lower value of $\beta$ for a group (and equivalently a low value of $\gamma$) implies that the group is more flexible in their schedule than travelers with higher value of $\beta$. Thus, we call travelers with lower value of $\beta$ (or $\gamma$) as \textit{absolutely time-flexible} travelers with higher absolute flexibility in their schedule. The ratio $\beta/\alpha$ represents the relationship between the absolute value of time while traveling and the early arrival schedule delay penalty. We call travelers with lower value of $\beta/\alpha$  as \textit{relatively time-flexible} travelers. %
 These travelers have more flexibility in their schedule relative to the time they are willing to spend in the queue. Table \ref{tab:timePoor} summarizes the definitions and provides example usages.  Without loss of generality, we \vpNew{label the groups so that} group 1 is relatively more (or as) time-flexible as group 2, that is,

\begin{equation}
\beta_1/\alpha_1 \leq \beta_2/\alpha_2.
\label{eq:orderingOfGroups}
\end{equation}

\begin{table}[]
\centering
\caption{Characterizing travelers based on absolute and relative values of $\alpha$ and $\beta$}
\label{tab:timePoor}
\begin{tabular}{|l|c|}
\hline
\multicolumn{1}{|c|}{\textbf{Characteristics}} & \textbf{Label} \\ \hline
\multicolumn{1}{|c|}{High value of $\alpha$} & High-income travelers \\ \hline
\multicolumn{1}{|c|}{Low value of $\beta$ (or $\gamma$)} & Absolutely time-flexible travelers \\ \hline
\multicolumn{1}{|c|}{Low value of $\beta/\alpha$} & Relatively time-flexible travelers \\ \hline
\multicolumn{2}{|l|}{\cellcolor[HTML]{C0C0C0}} \\ \hline
\multicolumn{2}{|l|}{\textbf{Example usages:}} \\ \hline
\multicolumn{2}{|l|}{1) Group 1 is absolutely less time-flexible than group 2 meaning $\beta_1 >\beta_2$} \\ \hline
\multicolumn{2}{|l|}{2) Group 1 is relatively less time-flexible than group 2 meaning $\beta_1/\alpha_1 > \beta_2/\alpha_2$} \\ \hline
\multicolumn{2}{|l|}{3) Group 1 is relatively as time-flexible as group 2 meaning $\beta_1/\alpha_1 = \beta_2/\alpha_2$} \\ \hline
\multicolumn{2}{|l|}{\begin{tabular}[c]{@{}l@{}}4) Group 1 is absolutely more time-flexible but relatively less time-flexible \\ than group 2 meaning $\beta_1 < \beta_2$ but $\beta_1/\alpha_1 > \beta_2/\alpha_2$\end{tabular}} \\ \hline
\end{tabular}
\end{table}

The ratio \vpNew{$\beta/\alpha$} models relative willingness to pay and is integral to capturing the elements of time poverty. %
For example, consider two travelers where the first traveler has $\alpha$ as $100$ units and $\beta$ as $10$ units, while the second traveler has $\alpha$ as $30$ units and $\beta$ as $9$ units (\vpNew{as in} the example in Section~\ref{sec:intro}). The first traveler has higher $\alpha$ and thus has a greater absolute willingness to pay. The first traveler also has higher $\beta$ and thus is absolutely less time-flexible. However, the first traveler's worth of arriving at their destination is 10\% of their absolute willingness to pay while the second's is 30\%.  This means that the second traveler who is relatively less time-flexible than the first traveler has higher relative willingness to pay.

There are other alternatives for modeling heterogeneous vehicle groups \vpNew{for instance,} considering that $\alpha$ and $\beta$  \vpNew{are uniform} groups and $\gamma_k/\beta_k$ is what varies, or that vehicles differ not in their $\alpha$, $\beta$, or $\gamma$ values, but in the desired arrival time at the destination, $\tau^*$~\citep[][]{arnott1987schedule}. In this article, we focus our attention on the case where assumption A$\#8$ holds true as under this case schedule delays are not minimized at user-equilibrium and thus a system optimal toll that minimizes the net costs of the system has a potential to favor absolutely time-flexible travelers over relatively time-flexible travelers, which we show in Section~\ref{subsec:SOtoll} creates inequity.

	\subsection{No-toll equilibrium}

	The no-toll equilibrium was first derived by~\cite{arnott1987schedule}. The key result under the stated assumptions can be summarized as \vpNew{follows}: ``at equilibrium, a fraction $\eta/(1+\eta)$ of each group departs early, with the remainder departing late. For early departures, the traveler with lowest value of $\beta/\alpha$ travel first, followed by other travelers in the increasing order of $\beta/\alpha$ values. For late departures, travelers with highest value of $\beta/\alpha$ travel first, followed by other travelers in the increasing order of $\beta/\alpha$ values. If the inequalities [in Equation~\eqref{eq:orderingOfGroups}] are strict, there is no overlap between departure times of two groups." Thus, under the presence of no-toll, travelers who are relatively least time-flexible depart and arrive closest to the desired arrival time.
	
	The above result can be used to determine the first and the last times of departure, denoted by $t_0$ and $t_f$, respectively. We denote the transition time when the early travelers from group 1 cease traveling and the early travelers from group 2 begin as \vpNew{$t_{A}$}. Similarly, \vpNew{$t_{B}$} is the time when the late travelers from group 2 cease traveling and the late travelers from group 1 begin. We add a superscript ``NO" to indicate that these times are derived for no-toll case. The expression for transition departure times can be derived a function of $\tau^*$ as follows:
	
	\begin{align}
	t_0^{\text{NO}} &= \tau^* - \frac{\eta}{1+\eta} \frac{N}{D} \\
	t_f^{\text{NO}} &= \tau^* + \frac{1}{1+\eta} \frac{N}{D} \\
	t_{A}^{\text{NO}} &= \tau^* - \frac{\eta}{\eta+1}\frac{N_2}{D} - \frac{\beta_1}{\alpha_1} \frac{\eta}{1+\eta} \frac{N_1}{D} \\
	t_{B}^{\text{NO}} &= \tau^* + \frac{1}{\eta+1}\frac{N_2}{D} - \frac{\beta_1}{\alpha_1} \frac{\eta}{1+\eta} \frac{N_1}{D} \label{eq:timeLocation3}
\end{align}

	These departure time values can be used to compute cost expressions which are shown in detail in Appendix~\ref{appendix:costs}. We provide numerical comparison of costs in Section~\ref{sec:numerical}.

\subsection{System-optimal toll and equity concerns}
\label{subsec:SOtoll}
	A system-optimal (SO) toll minimizes the total cost across the system (excluding the toll costs as these are assumed external to the benefits). This is achieved by charging a uniform toll for all travelers who decide to depart at $t$, that is $c_{\text{toll}}^{1,\text{SO}}(t)=c_{\text{toll}}^{2,\text{SO}}(t) $.
	
	Queuing costs are minimized when travelers arrive at the discharge rate and there is no queue. That is, $c_{\text{queue}}^{k,\text{SO}}(t)=0$ for all $k\in \{ 1,2\}$ and $n(t)=D$. Because travelers face no queue, the arrival and departure times are identical, that is $\tau(t)=t$. The rush hour begins and ends at same time as the no-toll equilibrium as the bottleneck is fully utilized and the values of total number of travelers and discharge rate is identical (that is, $t_0^{\text{SO}}=t_0^{\text{NO}}$ and $\tau(t_f^{\text{SO}})=\tau(t_f^{\text{NO}})$). Schedule delay costs are minimized when the group which is relatively less time-flexible travels closest to $\tau^*$. Thus, we seek tolls such that the travelers who are absolutely less time-flexible travel closer to $\tau^*$.

  At equilibrium all travelers segregate their departure times such that no traveler from a same group travels together with a traveler from another group (unless the $\beta$ values are identical for the groups). Each travel group has an early departure and late departure. During early departure travelers from group with lower $\beta$ leave first, followed by travelers with higher $\beta$. During late departure, travelers from group with higher $\beta$ (and thus higher $\gamma$) travel first, followed by travelers of group with lower $\beta$ (or lower $\gamma$). Additionally, similar to the no-toll equilibrium case, the proportion of travelers traveling early is $\eta/(1+\eta)$, which is same across all groups~\citep{arnott1987schedule}.
  
  The equilibrium property also implies that toll rates increase at the rate $\beta_k$ when group $k$ is traveling early and decrease at the rate $\gamma_k$ when group $k$ is traveling late. For deriving the values of transition departure times, the SO toll has two scenarios.
  
  \begin{itemize}
  	\item Case 1: relatively more time-flexible travelers also are absolutely more time-flexible ($\beta_1/\alpha_1 < \beta_2/\alpha_2$ and $\beta_1 < \beta_2$). Under this case, the order of travel is same as the no-toll equilibrium. That is, 
  	\begin{equation}
	t_{A}^{\text{SO}} = \tau(t_{A}^{\text{SO}}) = \tau^* - \frac{\eta}{\eta+1}\frac{N_2}{D}
	\end{equation}
	\begin{equation}
	t_{B}^{\text{SO}} = \tau(t_{B}^{\text{SO}}) = \tau^* + \frac{1}{\eta+1}\frac{N_2}{D}
	\end{equation}

  	\item Case 2: relatively more time-flexible travelers are absolutely less time-flexible ($\beta_1/\alpha_1 < \beta_2/\alpha_2$ and $\beta_1 > \beta_2$). This case can only happen when $\alpha_1 > \alpha_2$, that is travelers from group 1 have higher income than travelers from group 2. Under this case, the order of travel under SO toll reverses (group 2 travels first, followed by group 1, then followed again by group 2). The transition departure times formulae simply replace $N_2$ with $N_1$. That is, 
  	\begin{equation}
	t_{A}^{\text{SO}} = \tau(t_{A}^{\text{SO}}) = \tau^* - \frac{\eta}{\eta+1}\frac{N_1}{D}
	\end{equation}
	\begin{equation}
	t_{B}^{\text{SO}} = \tau(t_{B}^{\text{SO}}) = \tau^* + \frac{1}{\eta+1}\frac{N_1}{D}
	\end{equation}
	Note that in this case $t_{A}^{\text{SO}}$ is the time when travelers who depart early (group 2) cease to travel and the other group (group 1) starts their travel. The order of the groups is reversed, but we use the same notation to warrant consistency.
  \end{itemize}
  
  \vpNew{The reversal} of the order of travel under case 2 is \vpNew{one} reason why the SO toll is inequitable. High-income travelers with a higher value of $\beta$ and $\gamma$ travel at the center of the peak hour regardless of the relative weight of their schedule delay penalty and groups with a lower value of time but higher relative \vpNew{weight} are shifted to the margins of the peak period. Thus, on top of having their actual costs increased by the toll, low-income travelers are further harmed by being made more likely to experience additional time poverty. Leisure time, as discussed in Section~\ref{subsec:timePoverty}, is a primary social good. The opportunity to eat breakfast with one’s kids or head home to pick them up from soccer practice is as valuable to low-income families as high income families in relative terms, even though low-income families are not in a position to spend money in a way that reveals this preference on the toll road. 
  
  In the following sections, we focus our attention on case 2. Under this case, the SO toll can be expressed as Equation \eqref{eq:SOtoll}, where we use the expressions for $t_{A}^{\text{SO}}$ and $t_{B}^{\text{SO}}$ to cancel terms $\beta_1(\tau^* - t_{A}^{\text{SO}})$ and $ \gamma_1 (t_{B}^{\text{SO}} - \tau^*)$. The new expressions for costs $SDC$, $TRC$, and $TC$ under SO toll for this case are detailed in Appendix~\ref{appendix:costs}.
 
 \begin{equation}
 	c_{\text{toll}}^{1,\text{SO}}(t)=c_{\text{toll}}^{2,\text{SO}}(t) =
 	\begin{cases}
 	 0 \qquad &t \leq t_0^{\text{SO}} \\
 	 \beta_2(t - t_{0}^{\text{SO}}) \qquad &t_0^{\text{SO}} < t \leq t_{A}^{\text{SO}} \\
 	 \beta_2(t_{A}^{\text{SO}} - t_{0}^{\text{SO}}) + \beta_1(t - t_{A}^{\text{SO}}) \qquad &t_{A}^{\text{SO}} < t \leq \tau^{*} \\
 	 \beta_2(t_{A}^{\text{SO}} - t_{0}^{\text{SO}}) + \beta_1(\tau^* - t_{A}^{\text{SO}})- \gamma_1 (t - \tau^*) \qquad &\tau^* < t \leq t_{B}^{\text{SO}} \\
 	 \beta_2(t_{A}^{\text{SO}} - t_{0}^{\text{SO}}) -\gamma_2(t-t_{B}^{\text{SO}}) \qquad &t_{B}^{\text{SO}} < t \leq t_{f}^{\text{SO}} \\
 	 0 \qquad & t_f \leq t
 	\end{cases}
 	\label{eq:SOtoll}
\end{equation}

  	The second reason why the SO toll is inequitable is that it leads to welfare differences across different travelers. Focusing on case 2 where the order of departure is reversed relative to the no-toll case, SO tolls benefit travelers with higher value of $\alpha$. Assuming that the toll revenues are not rebated directly to the travelers (A$\#9$)\footnote{\cite{litman1997evaluating} argues that this is a reasonable assumption: because road networks are already subsidized by taxes on non-drivers and through various externalities, the toll-payers are not even justified in demanding that the rebate be given to them as a group. Instead it ought to be spent on social priorities that benefit all of the public, such as water infrastructure, mass transit or even parks.}, \cite{arnott1987schedule} show that under this case the costs of group 1 is lower while the costs of group 2 is higher relative to the no-toll case. That is, group 1 travelers benefit at the expense of group 2. We will quantify these differences in Section~\ref{sec:numerical}.
  	
  	The SO tolls are horizontally equitable as everyone faces an identical option to pay the face value of the toll and can choose to travel whenever they believe they will experience a minimal cost. However, it can lead to reversal of order of travel and inequitable distribution of welfare benefiting income-rich travelers from group 1 at the expense of group 2. In the next section, we propose a vertically-equitable tolling scheme which addresses \vpNew{these concerns}.

\section{Time-equitable Toll}
\label{sec:timeEquitableToll}

\subsection{Definition}
The time inequity of the SO tolling scheme is a feature, not a bug of that system. Low-income travelers, who are already more likely to experience time poverty, are shunted to the margins of the peak travel period in an effort to decrease the sum of the costs of travel for all citizens. Wealthier travelers with a higher $\alpha$ experience a higher queuing cost for a given level of congestion (by definition) and so the efforts to decrease total system costs have large marginal value from reducing the costs for wealthier travelers. A time-equitable (TE) toll addresses the issues with the SO toll by meeting three key desiderata.%

\begin{enumerate}
	\item \textbf{Preserves order}: The tolling scheme should preserve the same order of departures as in the no-toll scenario. In the no-toll case, the departure order is a reflection of the relative weight a group places on their schedule delay. An equitable program would show respect for all travelers by preserving this ordering of preferences rather than attempting to coerce travelers to choose \vpNew{departure} times based on their absolute value of time.
	\item \textbf{Maintains zero queue}: Similar to the SO toll, a TE toll must have zero queue at equilibrium. This requirement is fundamental, \vpNew{reflecting that} a traveler paying a toll should be rewarded with no congestion. Thus, $c_{\text{queue}}^k(t)=0$, $n(t)=D$, and $\tau(t)=t$ for all $k\in K$ and $t$. Furthermore, since the bottleneck is always used at full capacity and the demand remains fixed, the earliest and last arrival times under TE tolls match that of the no-toll case. That is, $t_0^\text{TE} = t_0^\text{NO}$ and $\tau(t_f^\text{TE}) = \tau(t_f^\text{NO})$.
	\item \textbf{Revenue-neutral}: A time-equitable toll should generate the same revenues as the system optimal toll. As discussed earlier in Section~\ref{subsec:timePoverty}, cities and states use tolls both to ease congestion and to shift the burden of paying for infrastructure onto those who use it. Alternatives to the system optimal toll should raise at least as much money so that it does not shift the cost of the infrastructure back onto non-travelers. The idea is to make travel at the middle of the peak period affordable and fair, not to make it ``free".
\end{enumerate}

 We focus our attention on case 2 where absolutely more time-flexible travelers are relatively less time-flexible. That is, $\beta_1 > \beta_2$ and $\beta_1/\alpha_1 < \beta_2/\alpha_2$. As before, $\gamma/\beta = \eta$ is assumed constant across all groups. If $\beta_1 < \beta_2$ and $\beta_1/\alpha_1 > \beta_2/\alpha_2$, then the SO toll satisfies all three requirements and is thus time-equitable.
 
 \subsection{Deriving time equitable toll}

Meeting the stated requirements of a time equitable toll is best achieved by charging different rates to travelers from different groups. That is, we seek tolls that are vertically equitable with respect to relative willingness to pay. %
Such a structure is acceptable because we want to favor the disadvantaged groups and compensate for the inefficiencies of the SO toll. We defer the discussion on real-world implementation of such vertically-equitable tolls to Section~\ref{subsec:TE2tollOtherConcerns}.

We propose a toll for each group that scales a function, which is \vpNew{uniform across groups}, by the value of time \vpNew{for that}  group. The SO tolls value the worth of a traveler based on their schedule delay penalty parameters and thus lead to unjust switches between the order of departures. A toll that scales with the $\alpha$ value of each group prevents the wealthier travelers who are absolutely less time-flexible to receive benefits at the expense of poor travelers who are relatively less time-flexible.

\vpNew{Let $\zeta(t)$ be this function which is uniform across groups. The actual toll paid by a traveler is then $c_{\text{toll}}^k(t)= \alpha_k \zeta(t)$. We desire a value of $\zeta(t)$ such that the equilibrium condition in Definition~\ref{def:equil} is satisfied, which we derive using isocost-$\zeta$ curves.} 

\vpNew{For a moment, assume $\zeta(t)$ can be different across groups. If the total cost for all travelers in group $k$ is $C_k$, then an isocost-$\zeta$ curve for group $k$, denoted by $\zeta_k(t_a \mid C_k)$,  represents the variation in $\zeta(t_a)$ needed for the total cost to be $C_k$ as a function of destination arrival time $t_a$.  Mathematically, $\zeta_k(t_a \mid C_k) = \left( C_k - c_{\text{SD}}^k(t_a) \right)/\alpha_k$, which is equivalent to Equation~\eqref{eq:zetaIsocost}:} 

\begin{equation}
	\zeta_k(t_a \mid C_k) =
		\begin{cases}
			\left( C_k - \beta_k(\tau^*-t_a)\right)/\alpha_k & t_a<\tau^* \\
			\left(C_k - \gamma_k(t_a-\tau^*)\right)/\alpha_k & t_a \geq \tau^*
		\end{cases}.
	\label{eq:zetaIsocost}
\end{equation}

Figure \ref{fig:isocostZeta} shows the plots of isocost-$\zeta$ curves for both groups for varying values of $C_k$. As observed, the curve is an increasing function of arrival time with a slope $\beta_k/\alpha_k$ until $t_a=\tau^*$ after which it decreases with slope a $\gamma_k/\alpha_k$. Because group 2 is relatively less time-flexible, its isocost-$\zeta$ curves have steeper slopes than group 1. For a larger $C_k$ the arrival times are more spread out and thus increasing the value of $C_k$ shifts the corresponding isocost-$\zeta$ curve parallely outwards. Let $C_k^{\text{eq}}$ be the total cost experienced by group $k$ at the equilibrium. We highlight isocost-$\zeta$ curves corresponding to the equilibrium cost in bold. 

	\begin{figure}[H]
		\centering
		\includegraphics[scale=0.5]{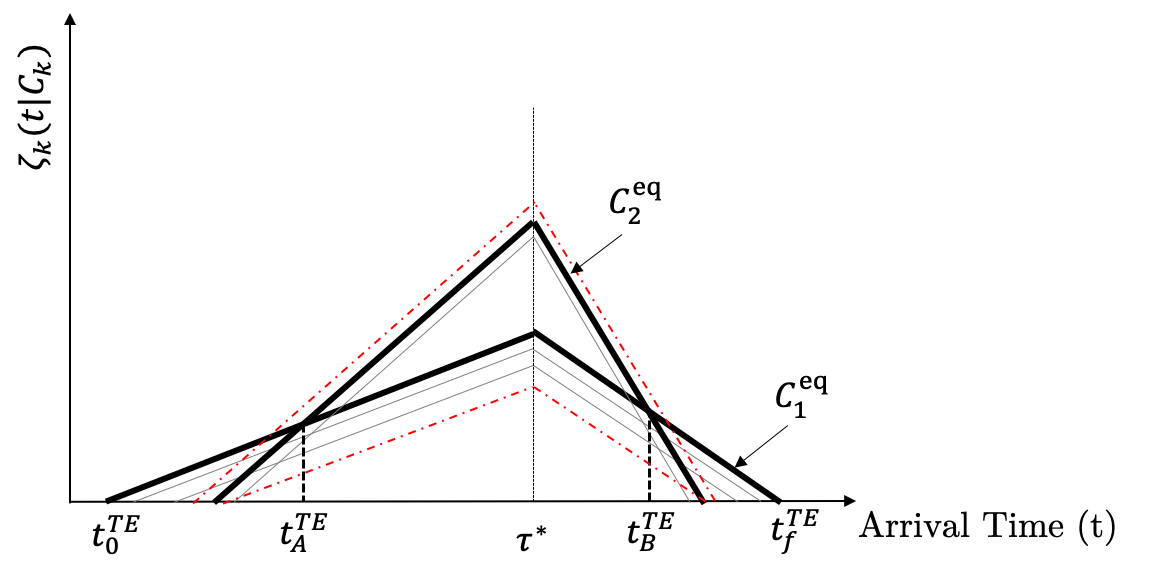}
		\caption{Isocost-$\zeta$ curves for groups $1$ and $2$. The bold curves correspond to the equilibrium costs $C_1^{\text{eq}}$ and $C_2^{\text{eq}}$ for the two groups, respectively. \vpNew{Because of Lemma~\ref{lem:zetaintersect}, the bold lines must intersect.}}
		\label{fig:isocostZeta}
	\end{figure}

At equilibrium, the experienced value of $\zeta(t)$ is same across all groups for any arrival time, that is $\zeta_k(t_a \mid C_k^\text{eq}) = \zeta(t_a)$ for all $k$ for every possible $t_a$. Given this, we argue that equilibrium isocost-$\zeta$ curves for both groups must intersect. 

\begin{lemma}
At equilibrium, isocost-$\zeta$ curves for both groups must intersect.
\label{lem:zetaintersect}
\end{lemma}
\begin{proof}
\vpNew{Let us assume to the contrary that, \vpNew{at equilibrium}, the curves do not intersect. For example, let the equilibrium isocost-$\zeta$ curves correspond to the red dash-dot pattern lines in Figure~\ref{fig:isocostZeta} where the isocost-$\zeta$ curve for group 2 is completely above the isocost-$\zeta$ curve for group 1. Because the experienced $\zeta(t)$ value is identical for both groups for a given arrival time, the travelers from group 2 can reduce their $\zeta$ values (and thus the total cost values) by switching to the arrival/departure time of travelers from group 1. This implies that the current isocost-$\zeta$ curve for group 2 is not at equilibrium, which is a contradiction.}
\end{proof}

We define $t_{A}^\text{TE}$ as the arrival time at the first intersection point and $t_{B}^\text{TE}$ as the arrival time at the second intersection point, as shown in Figure~\ref{fig:isocostZeta}.

All users in group $k$ must arrive in arrival times corresponding to the $C_k^{\text{eq}}$ curve. However, the arrival time of group $k$ cannot be any time $t_a$ when the isocost-$\zeta$ curve at equilibrium for group $k$ is below that for the other group $\bar{k}$, because otherwise the travelers from the group $\bar{k}$ can switch to arrival time $t_a$ and reduce their costs. Thus the arrival times for each group is such that the equilibrium value of $\zeta(t)$ is an upper envelope of the equilibrium isocost-$\zeta$ curves of both groups. That is, in Figure~\ref{fig:isocostZeta}, group $1$ arrives in time $(t_0^\text{TE},t_{A}^\text{TE}) \cup (t_{B}^\text{TE},t_f^\text{TE})$, while group $2$ arrives in time $(t_{A}^\text{TE},t_{B}^\text{TE})$. Given our assumption on group values, we thus have that the order of departure at equilibrium is the same as the no-toll scenario. Hence, tolls derived from the equilibrium  value of $\zeta(t)$ preserve the order of travel.

Similar to the arguments for no-toll, we can show that the number of travelers departing early for each group is a fixed proportion of the number of travelers in each group. 

\begin{thm}
	At equilibrium,  \vpNew{the proportion of each group that departs early is $\eta/(1+\eta)$.}
	\label{thm:propTravTEtoll}
\end{thm}
\begin{proof}
	See Appendix~\ref{appen:proofThm1}.
\end{proof}

The values of $t_{A}^\text{TE}$ and $t_{B}^\text{TE}$ are also derived in Appendix B and they have the same value as the arrival times for the no-toll case. The value of $\zeta(t)$ at the equilibrium is given by:

\begin{equation}
	\zeta(t) =
		\begin{cases}
		 0 & t<t_0^\text{TE} \\
		 \frac{\beta_1}{\alpha_1}(t-t_0^\text{TE}) & t_0^\text{TE} \leq t < t_{A}^\text{TE}\\
		 \frac{\beta_1}{\alpha_1}(t_{A}^\text{TE}-t_0^\text{TE}) + \frac{\beta_2}{\alpha_2}(t-t_{A}^\text{TE}) & t_{A}^\text{TE} \leq t < \tau^* \\
		 \frac{\beta_1}{\alpha_1}(t_{A}^\text{TE}-t_0^\text{TE}) + \frac{\beta_2}{\alpha_2}(\tau^*-t_{A}^\text{TE}) - \frac{\gamma_2}{\alpha_2}(t-\tau^*) & \tau^* \leq t < t_{B}^\text{TE} \\
		 \frac{\beta_1}{\alpha_1}(t_{A}^\text{TE}-t_0^\text{TE}) + \frac{\beta_2}{\alpha_2}(\tau^*-t_{A}^\text{TE}) - \frac{\gamma_2}{\alpha_2}(t_{B}^\text{TE}-\tau^*) - \frac{\gamma_1}{\alpha_1}(t-t_{B}^\text{TE}) & t_{B}^\text{TE} \leq t < t_{f}^\text{TE} \\
		 0 & t_f^\text{TE} \leq t
		\end{cases}
\end{equation}

We refer the tolls derived from this $\zeta(t)$ as $\text{TE1}$ tolls. The toll charged for group $k$ is given by $c_{\text{toll}}^{k,\text{TE1}}(t) = \alpha_k \zeta(t)$. We define $\wp$, called the toll-rate escalator, as the ratio of $\alpha/\beta$ for group 1 to that of group 2, that is, $\wp= \frac{\beta_2/\alpha_2}{\beta_1/\alpha_1}$. Given the assumption about the values of $\alpha_k$ and $\beta_k$, $\wp$ is always greater than 1. We can express $\text{TE1}$ tolls in terms of $\wp$ as in Equations \eqref{eq:TE1toll_1} and \eqref{eq:TE1toll_2}. We use Theorem \ref{thm:propTravTEtoll} for the time periods $t_{B}^\text{TE} < t \leq t_f^\text{TE}$ that allows $\beta_k (\tau^*-t_{A}^\text{TE}) = \gamma_k(t_{B}^\text{TE}-\tau^*)$ for $k\in \{1,2\}.$\footnote{$\text{LHS}=\beta_k (\tau^*-t_{A}^\text{TE})= \beta_k \left( \frac{\eta}{1+\eta} \frac{N_2}{D} \right) = \beta_k \eta \left( \frac{1}{1+\eta} \frac{N_2}{D} \right)= \gamma_k \left( \frac{1}{1+\eta} \frac{N_2}{D} \right) = \gamma_k(t_{B}^\text{TE}-\tau^*)=\text{RHS} $}

\begin{equation}
	c_{\text{toll}}^{1,\text{TE1}}(t) = 
	\begin{cases}
		\beta_1(t-t_0^\text{TE}) & t_0^\text{TE} \leq t \leq t_{A}^\text{TE} \\
		\beta_1(t_{A}^\text{TE}-t_0^\text{TE}) + \wp \cdot \beta_1(t-t_{A}^\text{TE}) & t_{A}^\text{TE} < t \leq \tau^* \\
		\beta_1(t_{A}^\text{TE}-t_0^\text{TE}) + \wp \cdot \left( \beta_1(\tau^*-t_{A}^\text{TE}) -\gamma_1(t-\tau^*)\right) & \tau^* < t \leq t_{B}^\text{TE} \\
		\beta_1(t_{A}^\text{TE}-t_0^\text{TE}) - \gamma_1(t-t_{B}^\text{TE}) & t_{B}^\text{TE} < t \leq t_{f}^\text{TE}
	\end{cases}
	\label{eq:TE1toll_1}
\end{equation}

\begin{equation}
	c_{\text{toll}}^{2,\text{TE1}}(t) =
	\begin{cases}
		\frac{1}{\wp} \beta_2 (t-t_0^\text{TE}) & t_0^\text{TE} \leq t \leq t_{A}^\text{TE} \\
		\frac{1}{\wp} \beta_2 (t_{A}^\text{TE}-t_0^\text{TE}) + \beta_2(t-t_{A}^\text{TE}) & t_{A}^\text{TE} < t \leq \tau^* \\
		\frac{1}{\wp} \beta_2 (t_{A}^\text{TE}-t_0^\text{TE}) + \beta_2(\tau^*-t_{A}^\text{TE}) - \gamma_2(t-\tau^*) & \tau^* < t \leq t_{B}^\text{TE} \\
		\frac{1}{\wp} \left( \beta_2 (t_{A}^\text{TE}-t_0^\text{TE}) - \gamma_2 (t-t_{B}^\text{TE}) \right) & t_{B}^\text{TE} < t \leq t_{f}^\text{TE}
	\end{cases}
	\label{eq:TE1toll_2}
\end{equation}

We note that the toll costs perfectly substitute for the queuing costs in the no-toll equilibrium case (as a function of the arrival time). Thus the total costs under the presence of $\text{TE1}$ tolls is same as the total costs under no-toll equilibrium and the revenue obtained from $\text{TE1}$ tolls is equal to the queuing costs under no-toll equilibrium. However, queuing costs under no-toll case are not equal to the revenue from SO tolls.\footnote{In fact, the analysis of cost expressions in Appendix~\ref{appendix:costs} shows that queuing costs under no-toll are lower than SO toll revenue.} Thus, $\text{TE1}$ tolls are not revenue-neutral. We next modify $\text{TE1}$ tolls such that the new set of tolls generate same revenue as the SO tolls.

We start with the following property of the tolls at equilibrium.

 \begin{lemma}
  If the demand is inelastic, the departure times for each group at equilibrium is preserved under following modifications to the $\text{TE1}$ toll: (a) adding a constant to the $\text{TE1}$ toll for all departure times, or (b) increasing the toll by any positive amount for a group $k$ for the duration when the group is not traveling.\label{lem:equilPreserve}
 \end{lemma}
 
 \begin{proof}
 Since the demand is inelastic, travelers have no other choice but to travel regardless of the toll value. We look at individual cases. First, if the tolls are increased uniformly across all time periods (including time periods outside the range $(t_0,t_f)$), then the order of departure and the departure times are  not impacted as travelers have no incentive to choose other departure times. Second, if the tolls are only increased in the time period when a traveler is not traveling and only for that group, then travelers are still comfortable with their current arrival times and have no incentive to switch to other arrival times, thus preserving the equilibrium departure times.
 \end{proof}
 
 Lemma \ref{lem:equilPreserve} provides ideas to modify $\text{TE1}$ tolls so that we achieve the revenue of SO tolls. %
 Let us parameterize the $\text{TE1}$ tolls in Equations \eqref{eq:TE1toll_1}--\eqref{eq:TE1toll_2} considering toll-rate escalator $\wp$ as a general parameter, resulting in a general toll $c_{\text{toll}}^{k,\text{Gen}}(t,\wp)$ for each group $k$. 
 
 We propose a set of tolls, called $\text{TE2}$ tolls, as follows:  $c_{\text{toll}}^{1,\text{TE2}}(t) = c_{\text{toll}}^{1,\text{Gen}}(t,\bar{\wp})$ for $\bar{\wp} \geq 1$, and $c_{\text{toll}}^{2,\text{TE2}}(t) = c_{\text{toll}}^{2,\text{Gen}}(t,1)$. Equations~\eqref{eq:TE2toll_1}--\eqref{eq:TE2toll_2} provide the expression. Next, we show that $\text{TE2}$ tolls satisfy desired properties of time-equitable tolls.

 \begin{equation}
	c_{\text{toll}}^{1,\text{TE2}}(t) = 
	\begin{cases}
		\beta_1(t-t_0^\text{TE}) & t_0^\text{TE} \leq t \leq t_{A}^\text{TE} \\
		\beta_1(t_{A}^\text{TE}-t_0^\text{TE}) + \bar{\wp} \cdot \beta_1(t-t_{A}^\text{TE}) & t_{A}^\text{TE} < t \leq \tau^* \\
		\beta_1(t_{A}^\text{TE}-t_0^\text{TE}) + \bar{\wp} \cdot \left( \beta_1(\tau^*-t_{A}^\text{TE}) -\gamma_1(t-\tau^*)\right) & \tau^* < t \leq t_{B}^\text{TE} \\
		\beta_1(t_{A}^\text{TE}-t_0^\text{TE}) - \gamma_1(t-t_{B}^\text{TE}) & t_{B}^\text{TE} < t \leq t_{f}^\text{TE}
	\end{cases}
	\label{eq:TE2toll_1}
\end{equation}

\begin{equation}
	c_{\text{toll}}^{2,\text{TE2}}(t) =
	\begin{cases}
		\beta_2 (t-t_0^\text{TE}) & t_0^\text{TE} \leq t \leq  \tau^* \\
		\beta_2 (\tau^*-t_0^\text{TE}) - \gamma_2(t-\tau^*) & \tau^* \leq t \leq t_f^\text{TE}
	\end{cases}
	\label{eq:TE2toll_2}
\end{equation}

\begin{thm}
	$\text{TE2}$ tolls preserve order of departure, maintain zero queue, and are revenue-neutral
\end{thm}

\begin{proof}
	For group 1, $\text{TE2}$ tolls are identical to $\text{TE1}$ tolls when a group departs early or late and only charges different values of tolls when a group does not travel. For group 2, $\text{TE2}$ tolls raise $\text{TE1}$ tolls uniformly by the amount $\left( 1- \frac{1}{\wp} \right) \beta_2 (t_{A}^\text{TE}-t_0^\text{TE})$ during the period group $2$ travels, and for all other times increases the toll arbitrarily. Thus, by Lemma \ref{lem:equilPreserve} the order of travel at equilibrium under $\text{TE2}$ tolls is same as $\text{TE1}$ tolls. Furthermore, because $\text{TE2}$ tolls are derived from $\text{TE1}$ tolls, they continue to maintain zero queue.

	What remains to show is that $\text{TE2}$ tolls generate same revenue as SO tolls, which we show using geometric arguments. The graphs of $\text{TE2}$ and $\text{SO}$ tolls are shown in Figure \ref{fig:SO_TE2_Rev} where the toll curve for each group is highlighted in solid lines of different thickness. Additionally, using the expressions of $t_{A}^{\text{SO}}$, $t_{A}^{\text{TE}}$, $t_{B}^{\text{SO}}$ and $t_{B}^{\text{TE}}$, the length of time scale can be expressed in terms of $N_k$ and $D$ as shown. Total revenue is the area under the bold curve multiplied by the discharge rate $D$.
	
	\begin{figure}[h]
		\centering
		\includegraphics[scale=0.7]{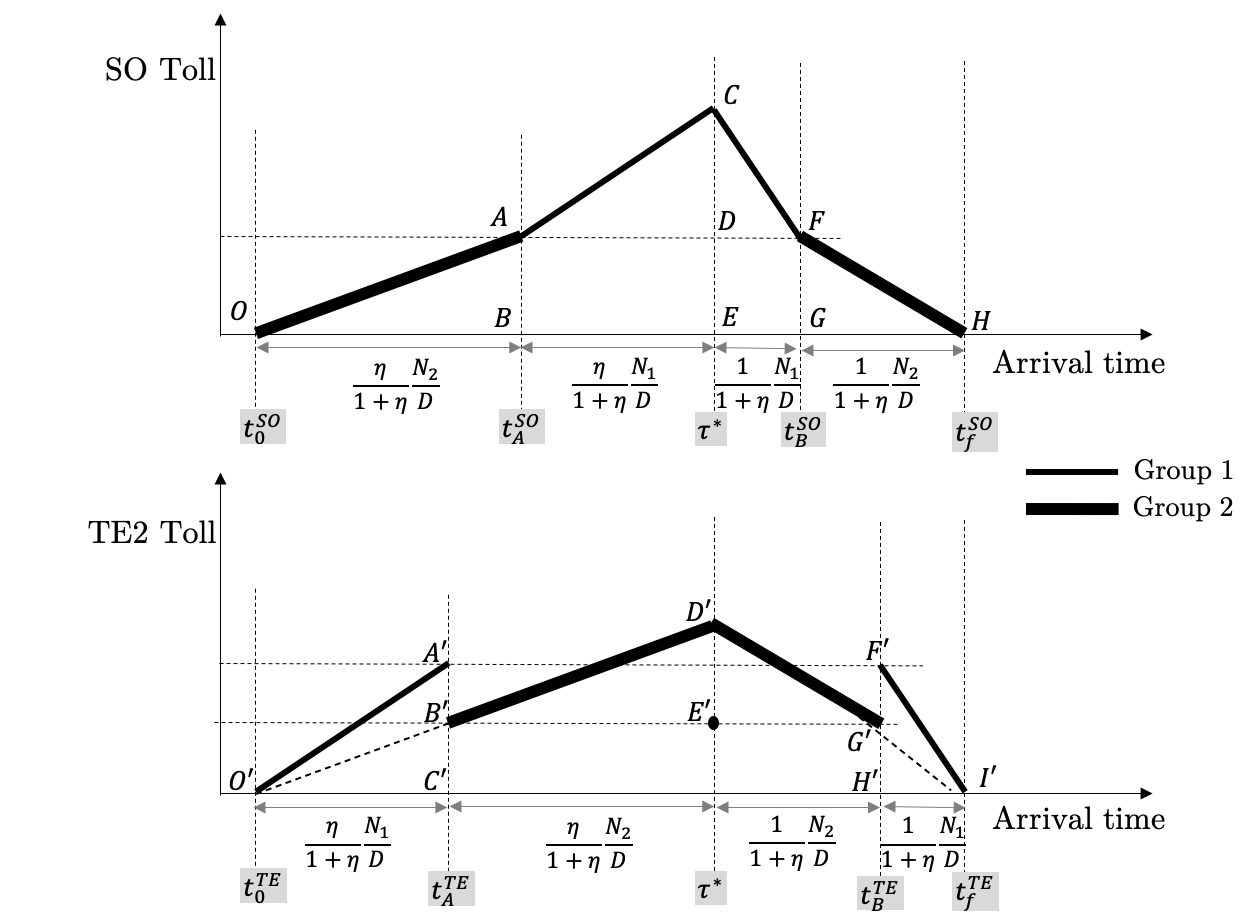}
		\caption{Toll curves for SO toll and $\text{TE2}$ toll for the two groups as a function of arrival times. The area under the curves is equal implying that both tolls generate same revenue}
		\label{fig:SO_TE2_Rev}
	\end{figure}
	
	We argue that the areas under the curve are equal. First, triangle $OAB$ in SO toll has the same area as triangle $B'E'D'$ in $\text{TE2}$ toll as both triangles are right angled and have the same base length and the slope of the hypotenuse. Similarly, the area of triangle $ACD$ is equal to the area of triangle $O'A'C'$, the area of triangle $CDF$ is equal to the area of triangle $F'I'H'$, and the area of triangle $FGH$ is equal to the area of triangle $D'E'G'$. The only remaining area is the area of rectangles $ABGF$ and $B'C'H'G'$ which are also equal as evaluated below:
	\begin{align*}
		\text{area}(ABGF) &= \text{length}(AB)*\text{length}(BG) \\
			&= \left( \frac{\beta_2 \eta}{1+\eta} \frac{N_2}{D}\right)*\left( \frac{N_1}{D} \right) \\
			&= \frac{\beta_2 \eta N_1 N_2}{(1+\eta)D^2}
	\end{align*}
	\begin{align*}
		\text{area}(B'C'H'G') &= \text{length}(B'C')*\text{length}(C'H') \\
			&= \left( \frac{\beta_2 \eta}{1+\eta} \frac{N_1}{D}\right)*\left( \frac{N_2}{D} \right) \\
			&= \frac{\beta_2 \eta N_1 N_2}{(1+\eta)D^2}
	\end{align*}
	
	Hence proved.
\end{proof}

We have established that $\text{TE2}$ tolls are time-equitable and generate the same schedule delay cost as the no-toll case and same revenue as the SO toll case. We include the complete cost expressions for the $\text{TE2}$ toll in the Appendix~\ref{appendix:costs}. There are other variants of time-equitable toll that can generate the same revenue as SO tolls; $\text{TE2}$ tolls are one example of this.

\subsection{Welfare analysis of $\text{TE1}$ and $\text{TE2}$ tolls}
In this section we compare the total cost for each group under $\text{TE1}$ and $\text{TE2}$ tolls relative to the case of no-toll. Under assumption A$\#9$, we consider that toll revenues are not rebated directly to the travelers (and are instead invested towards local community projects).

 For $\text{TE1}$ tolls, total costs for all travelers remain the same as the no-toll case. This is as expected, since the toll costs exactly substitute the travel time savings. However, for $\text{TE2}$ tolls, we show that the costs of group 2 are higher relative to the no-toll case, while the costs of group 1 are unimpacted.

\begin{prop}
	Under the case of no-rebate, costs of group 2 are higher under $\text{TE2}$ toll relative to the no-toll case, while the costs of group 1 remain the same.
	\label{prop:TE2tollworseoff}
\end{prop}

\begin{proof}
 Since the same proportion of group 1 travelers travel at the edge of the peak period in both the $\text{TE2}$ and the no-toll case, and since the start and end of peak period is the same, the cost of any traveler in group 1 is identical across both cases and equals $\frac{\beta_1 \eta}{1+\eta} \frac{N_1}{s}$. Thus, group 1 faces no extra cost relative to the no-toll case.
 
 Cost of each traveler in group 2 under $\text{TE2}$ toll is same across all departure times in $(t_{A}^{\text{TE}},t_{B}^{\text{TE}})$. Evaluating the cost of a traveler arriving at time $\tau^*$ (who only incurs the toll cost) we have,
 \begin{align}
 	TC_{2}^{\text{TE2}} &= \beta_2 (\tau^* - t_{0}^{\text{TE}}) \\
 	 &= \beta_2 \left( \frac{ \eta}{1+\eta} \frac{N}{s} \right).
 \end{align}
	Cost of each traveler in group 2 under no toll is obtained from~\cite{arnott1987schedule} as follows:
	\begin{equation}
		TC_{2}^{\text{NO}} = \beta_1 \left[ \frac{\alpha_2}{\alpha_1} - \left( \frac{\alpha_2}{\alpha_1} - \frac{\beta_2}{\beta_1} \right) f_2 \right] \left( \frac{ \eta}{1+\eta} \frac{N}{s} \right),
	\end{equation}	
	
	where $f_2$ is the proportion of group 2 travelers (that is, $f_2=N_2/N$). We compare $\beta_2$ with $\beta_1 \left[ \frac{\alpha_2}{\alpha_1} - \left( \frac{\alpha_2}{\alpha_1} - \frac{\beta_2}{\beta_1} \right) f_2 \right]$ and show that the latter expression is always smaller. This can be observed by the following set of inequalities:
	
 	\begin{align*}
 		& f_2 < 1 &(\because N_2 < N) \\
 		\Rightarrow \qquad & \left( \frac{\alpha_1}{\beta_1} - \frac{\alpha_2}{\beta_2} \right) f_2 \leq \left( \frac{\alpha_1}{\beta_1} - \frac{\alpha_2}{\beta_2} \right) & (\because \alpha_1/\beta_1 > \alpha_2/\beta_2) \\
 		\Rightarrow \qquad & \left( \frac{\alpha_1}{\beta_1} - \frac{\alpha_2}{\beta_2} \right) f_2 + \frac{\alpha_2}{\beta_2} \leq \frac{\alpha_1}{\beta_1} & \\
 		\Rightarrow \qquad & \beta_1 \left( \frac{\beta_2}{\beta_1} - \frac{\alpha_2}{\alpha_1} \right) f_2 + \frac{\alpha_2 \beta_1}{\alpha_1} \leq \beta_2 & (\text{multiply by }\beta_2\beta_1/\alpha_1) \\
 		\Rightarrow \qquad & \beta_1 \left[ \frac{\alpha_2}{\alpha_1} - \left(  \frac{\alpha_2}{\alpha_1}- \frac{\beta_2}{\beta_1} \right) f_2 \right] \leq \beta_2 &\\
 		\Rightarrow \qquad & TC_{2}^{\text{NO}}  \leq TC_{2}^{\text{TE2}}
 	\end{align*}
 Hence proved.
\end{proof}

Proposition \ref{prop:TE2tollworseoff} shows that, similar to SO toll, costs of group 2 are higher under $\text{TE2}$ toll relative to the no-toll equilibrium. This raises potential ``equity concerns." After all, if the point of the time-equitable toll is to be equitable, then a toll that eliminates congestion by charging relatively less time-flexible travelers (who have lower income) more than they believe their time waiting in traffic is worth while exactly substituting the cost of queuing for high income travelers with a toll is certainly suspicious. However, there are two ways in which the proposed tolling scheme is more equitable as becomes apparent when we compare the outcomes of the SO toll and the $\text{TE2}$ toll.

First, an increase in the cost of group $2$  under $\text{TE2}$ tolls are not obtained by making group $1$ better off. In both the SO and $\text{TE2}$ tolls group 2 experiences an increase in costs, but in the SO toll all of this welfare is transferred to group 1. Under $\text{TE2}$ tolls, although group 2 experiences higher costs, they are not compelled to transfer their costs into a benefit for other travelers. Should group 2 travelers decide that they do not have the funds to pay a toll in the center of the peak hour, they always have the option to forego the time benefits by choosing to travel outside of their group's time slot since their costs are flat over the entire peak hour.%

Second, because $\text{TE2}$ toll preserves the order, it prioritizes travelers who are relatively less time-flexible. $\text{TE2}$ tolls may not minimize the total costs for the system, yet, they preserve the order and generate vertical equity for time as a resource.

\section{Numerical Comparisons}
\label{sec:numerical}
In this section, we present numerical comparisons of costs under no toll, and the SO, $\text{TE1}$, and $\text{TE2}$ tolls. The analytic expressions of costs are provided in Appendix~\ref{appendix:costs}.

\subsection{Base case scenario}

First, we consider a base case scenario with the following parameter values. Consider two groups with 30 travelers each ($N_1=N_2=30$ and $N=60$). Let the values of $\alpha_1$ and $\alpha_2$ be $24 \aleph$ and $12\aleph$ \vpNew{per time unit}, respectively.\footnote{The $\aleph$ symbol is used to further underscore the break from empirical research and to guard against the temptation to see these values in terms of US dollars or other familiar currency.} We set the value of $\beta_1/\alpha_1 = 1/3$ and $\beta_2/\alpha_2=1/2$ to consider the case where low-income travelers in group 2 are relatively less time-flexible (since this is the case that highlights the time-poverty). We assume a constant value of the ratio of late and early arrival penalties as $4$ for both groups (that is, $\gamma_1/\beta_1 = \gamma_2/\beta_2=\eta = 4$). The bottleneck capacity is set as $6$ vehicles \vpNew{per time unit}. Without loss of generality, we consider desired arrival time $\tau^*=0$ for all vehicles.

First we compare the variation of different costs for each toll case with the arrival time at the destination as shown in Figure~\ref{fig:allTimeVaryingCostPlots}.

\begin{figure}
\centering
\subfloat[\label{fig:no_SD}]{\includegraphics[width=0.33\textwidth]{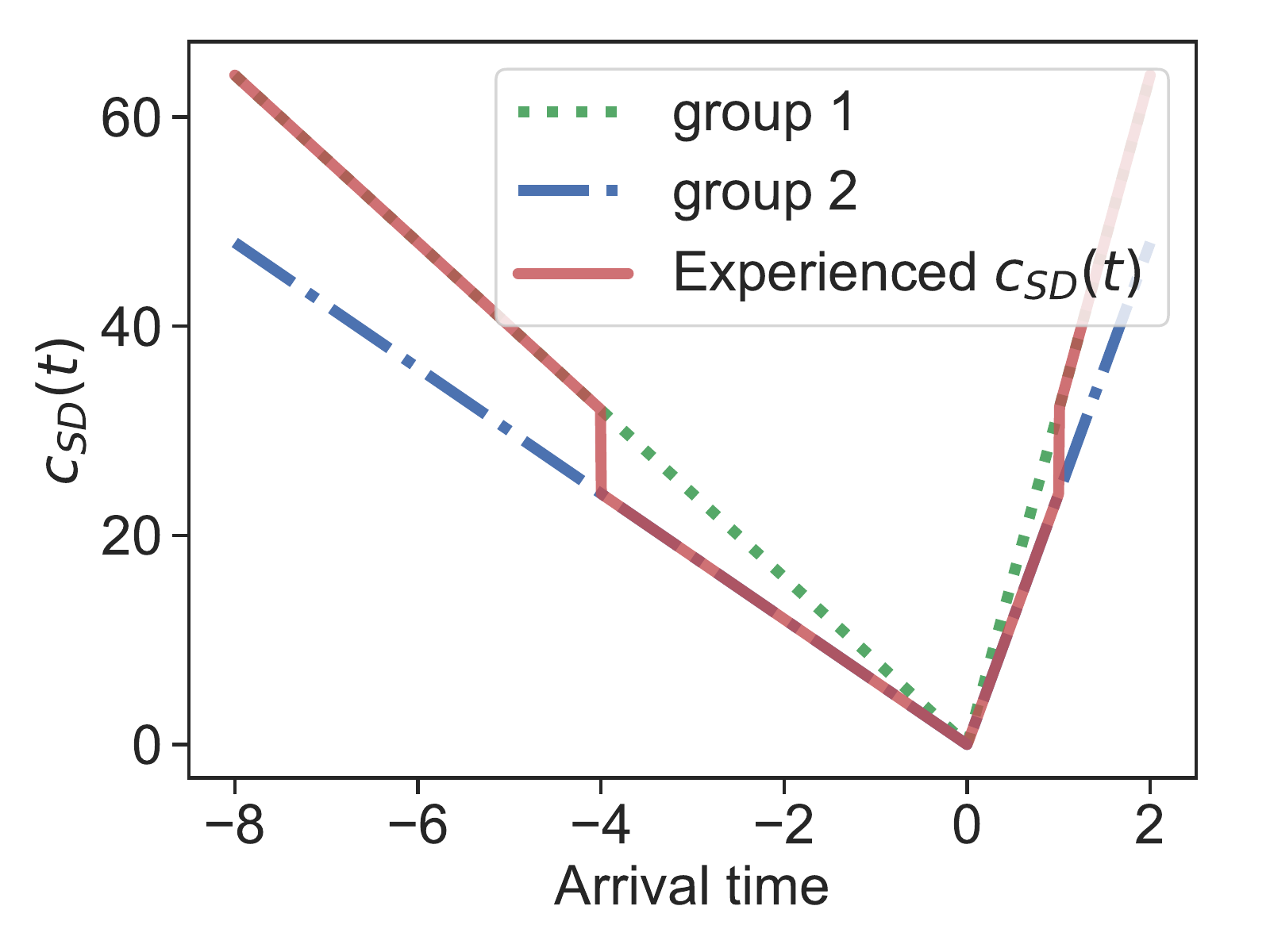}}\hfill
\subfloat[\label{fig:no_queue}] {\includegraphics[width=0.33\textwidth]{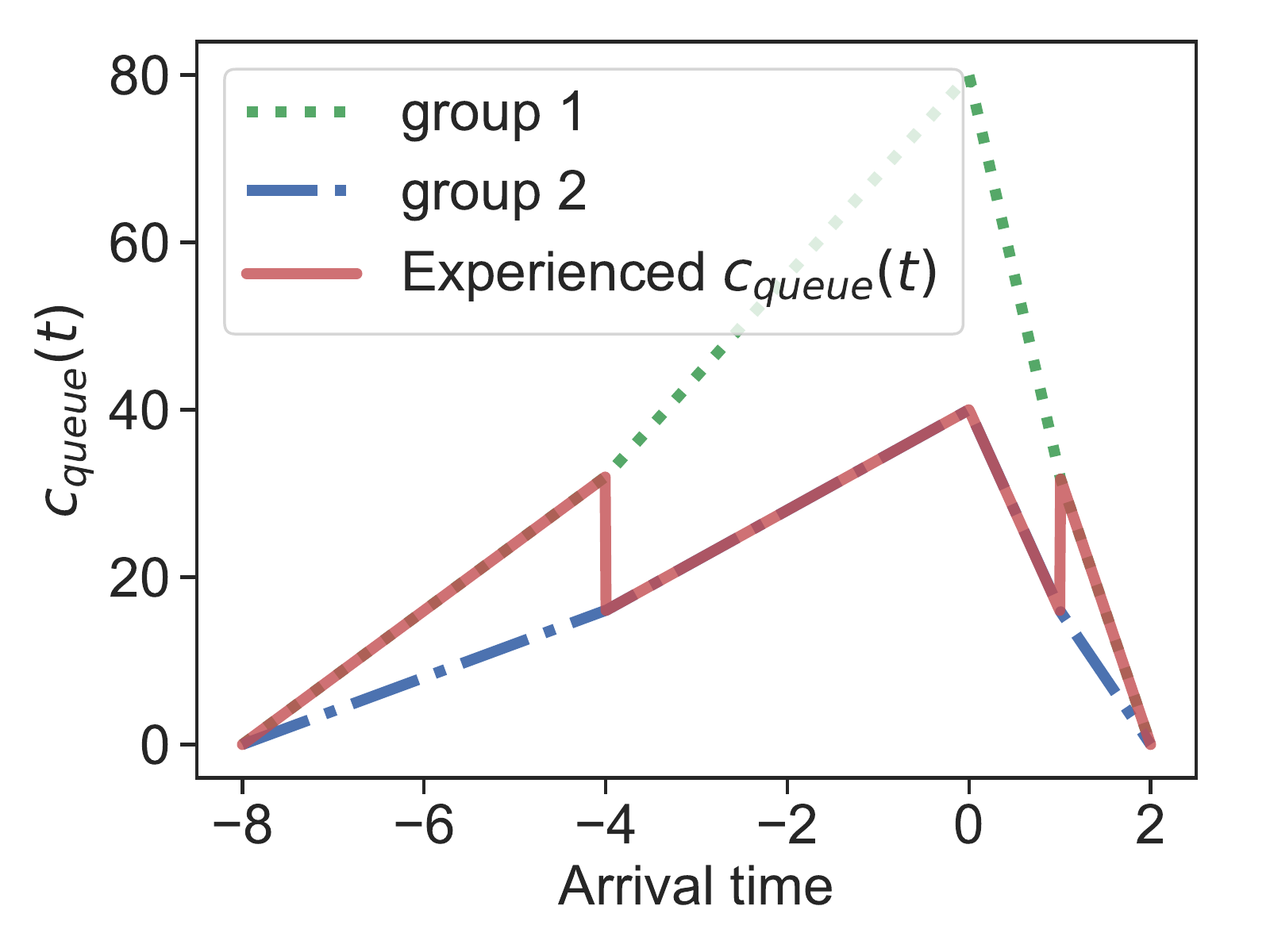}}\hfill
\subfloat[\label{fig:no_total}]{\includegraphics[width=0.33\textwidth]{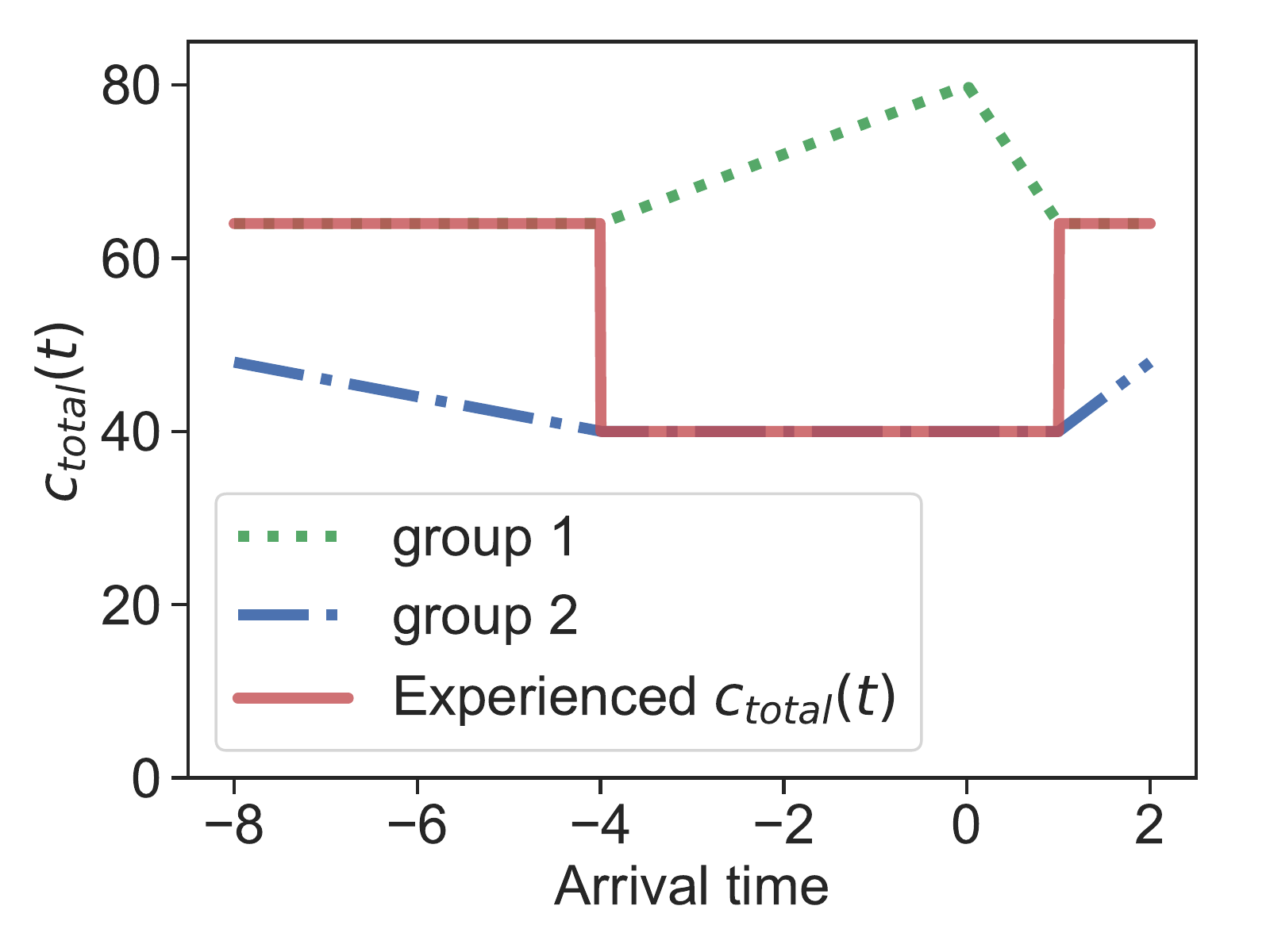}}\hfill
\subfloat[\label{fig:so_SD}]{\includegraphics[width=0.33\textwidth]{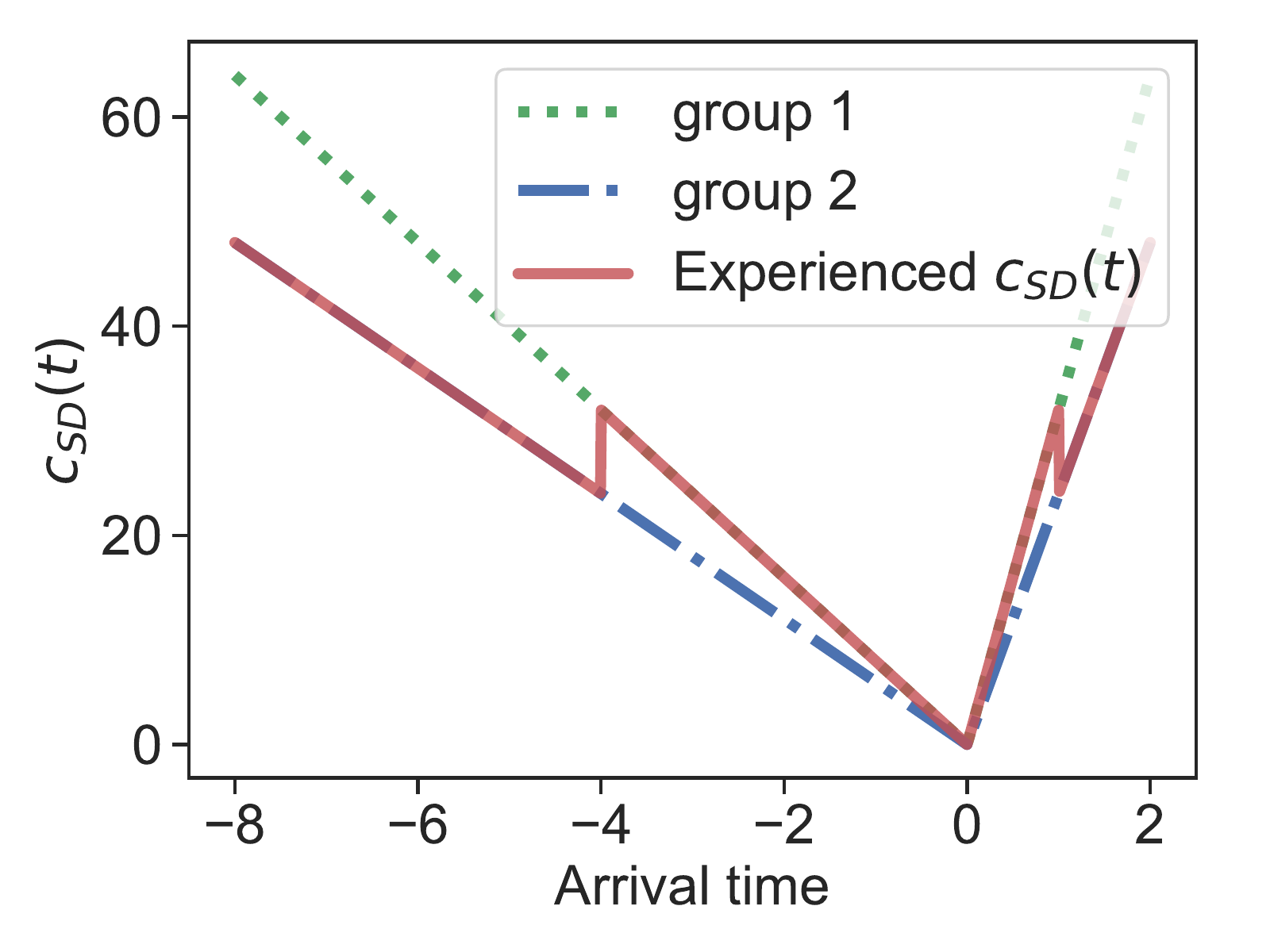}}\hfill
\subfloat[\label{fig:so_toll}] {\includegraphics[width=0.33\textwidth]{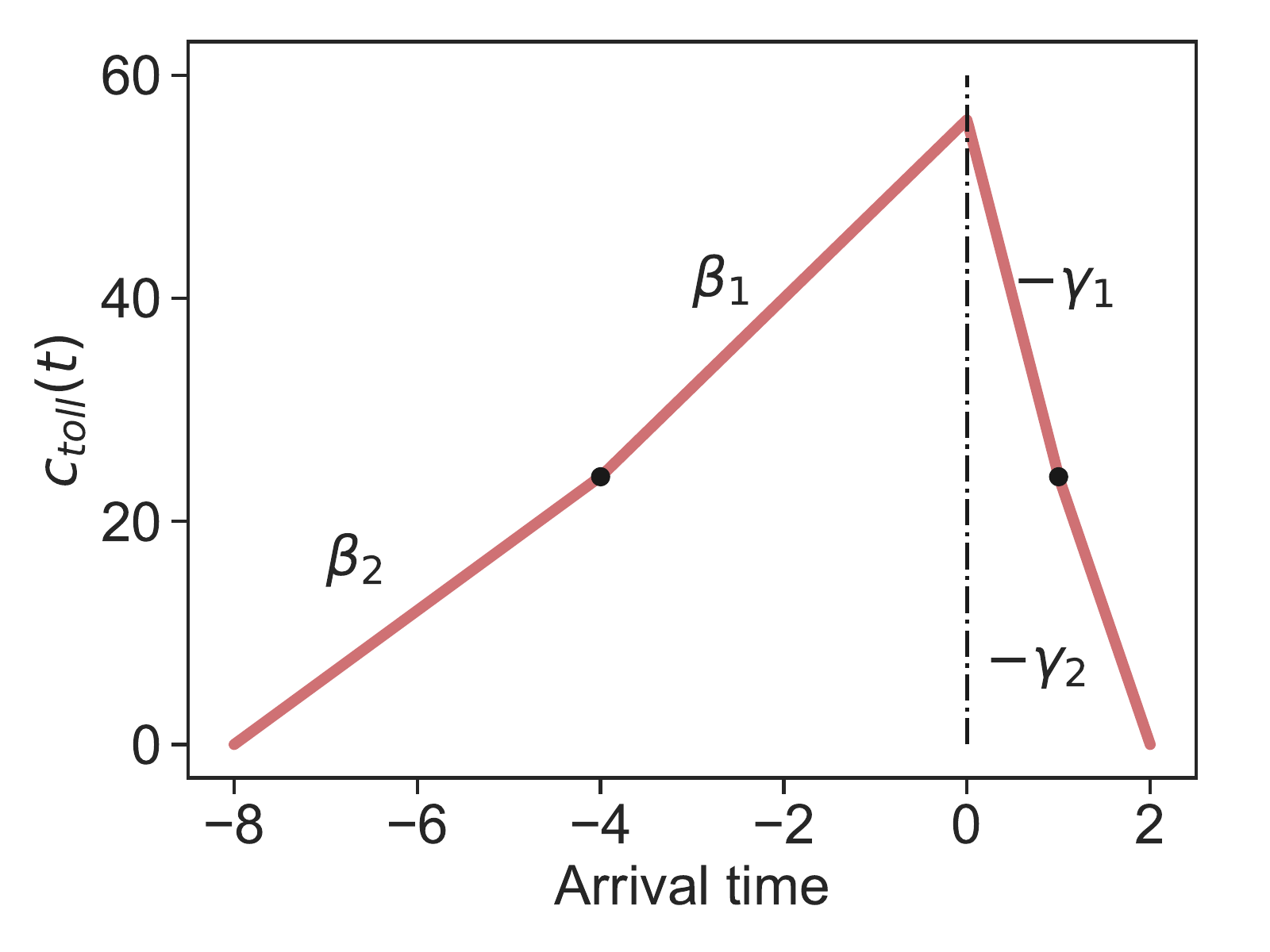}}\hfill
\subfloat[\label{fig:so_total}]{\includegraphics[width=0.33\textwidth]{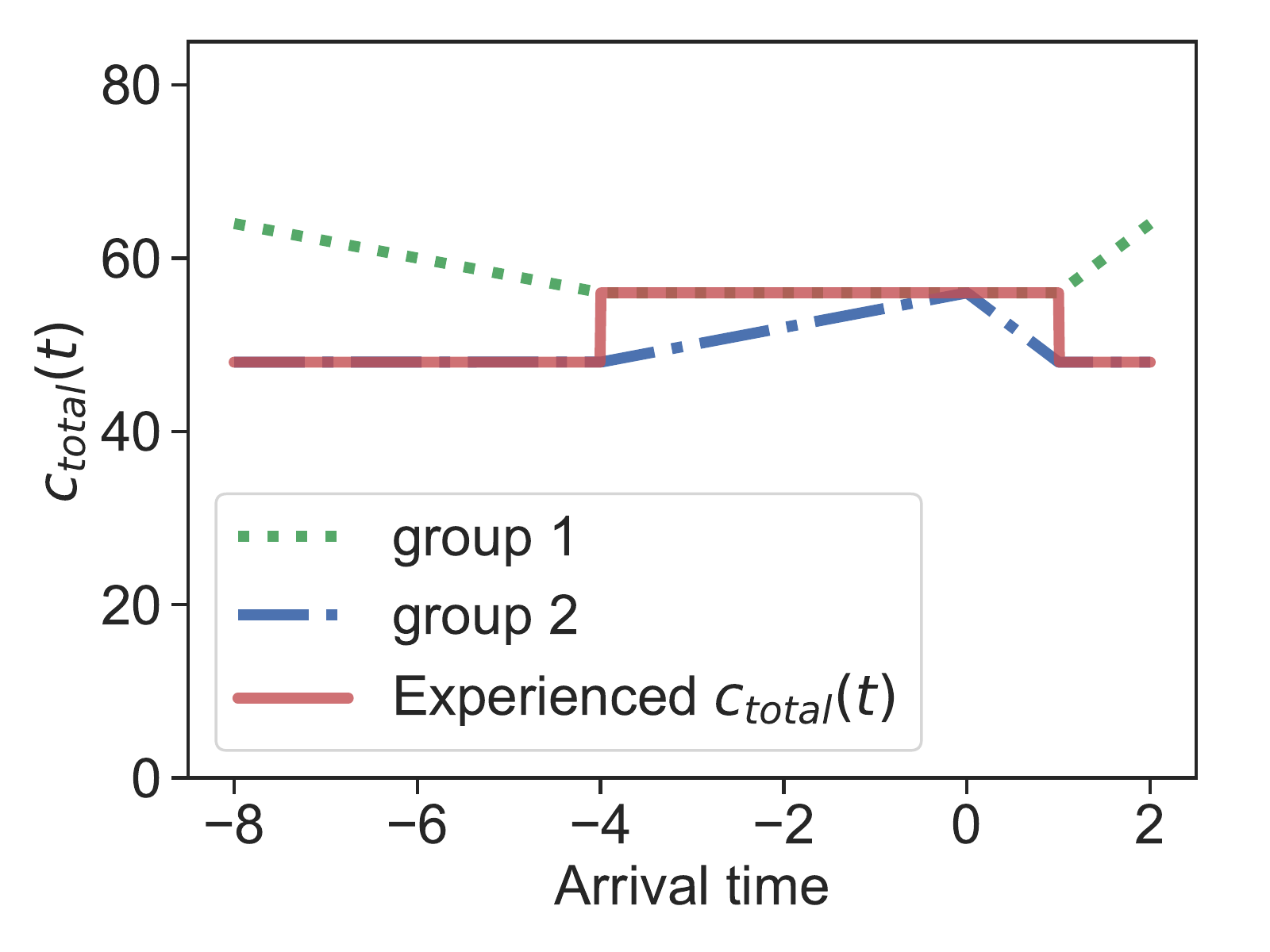}}\hfill
\subfloat[\label{fig:te1_sd}]{\includegraphics[width=0.33\textwidth]{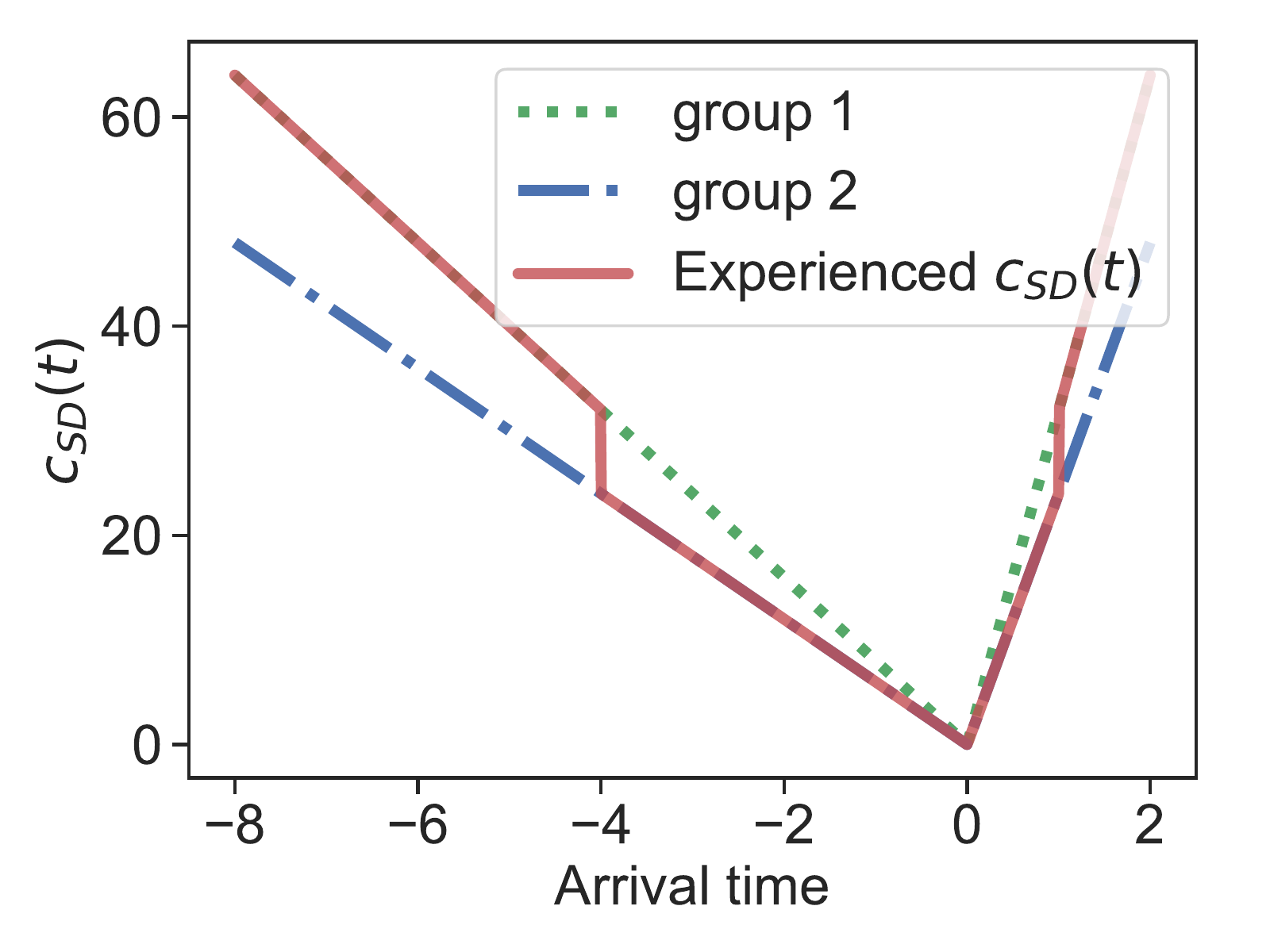}}\hfill
\subfloat[\label{fig:te1_toll}] {\includegraphics[width=0.33\textwidth]{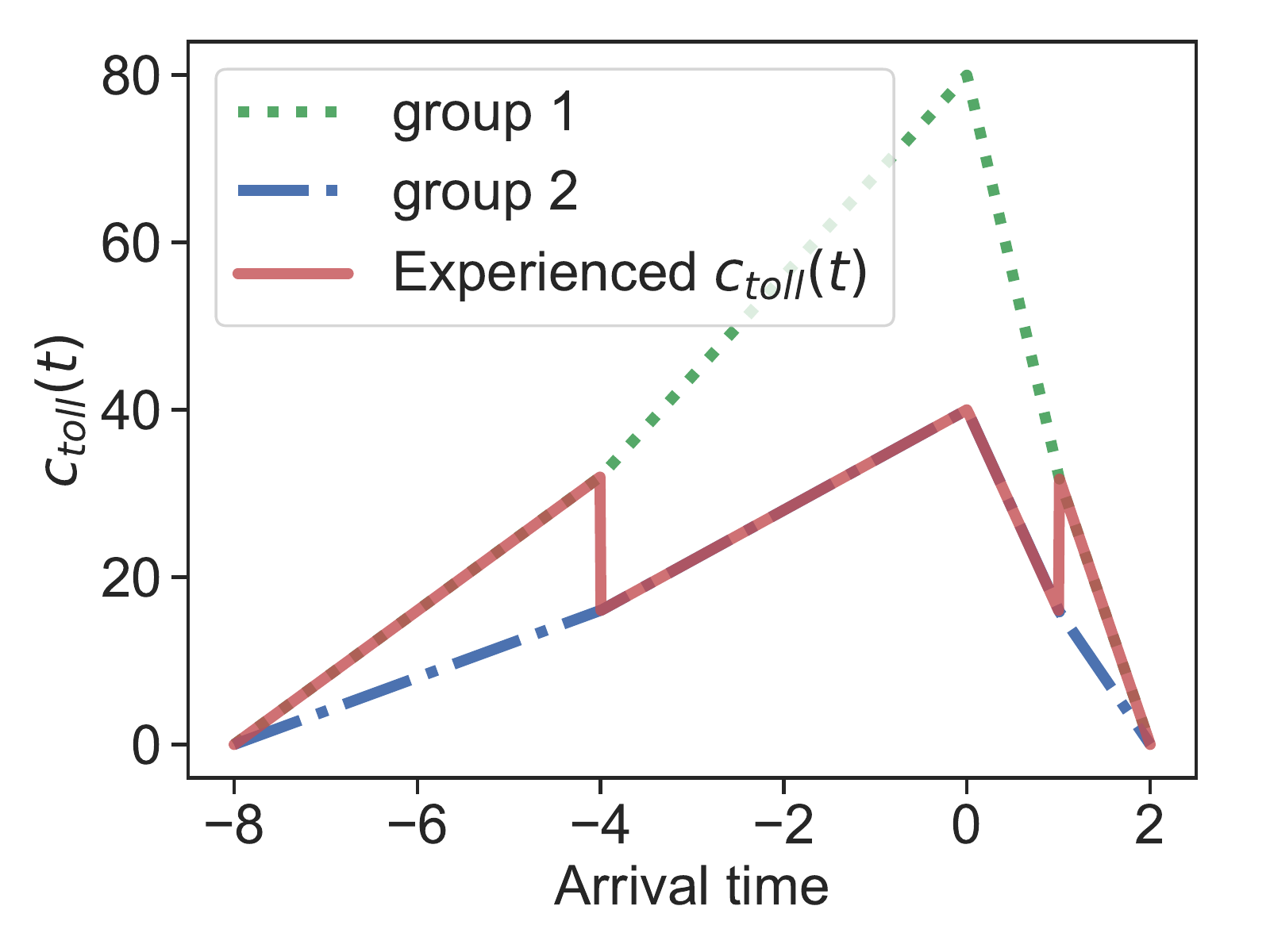}}\hfill
\subfloat[\label{fig:te1_total}]{\includegraphics[width=0.33\textwidth]{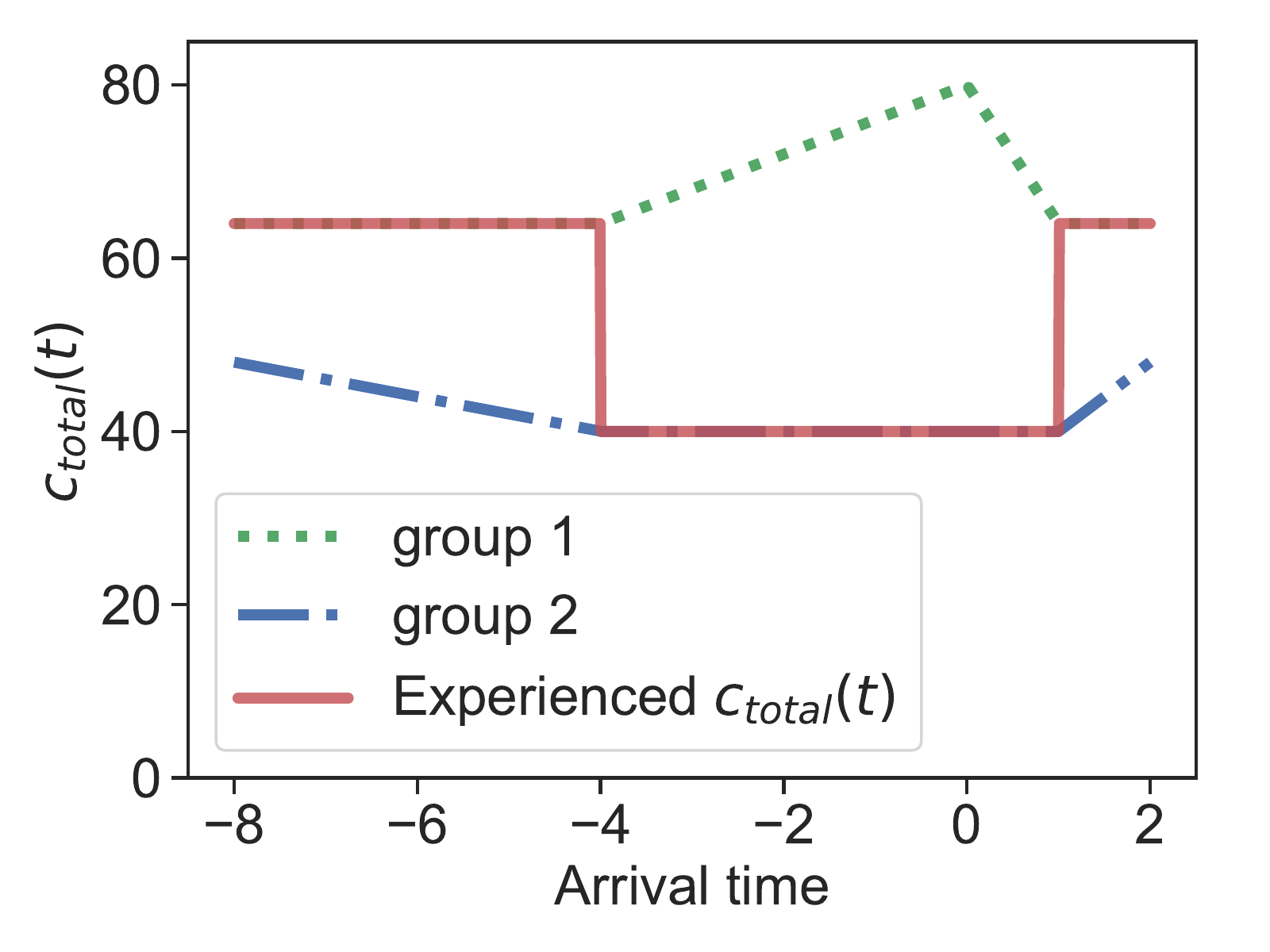}}\hfill
\subfloat[\label{fig:te2_sd}]{\includegraphics[width=0.33\textwidth]{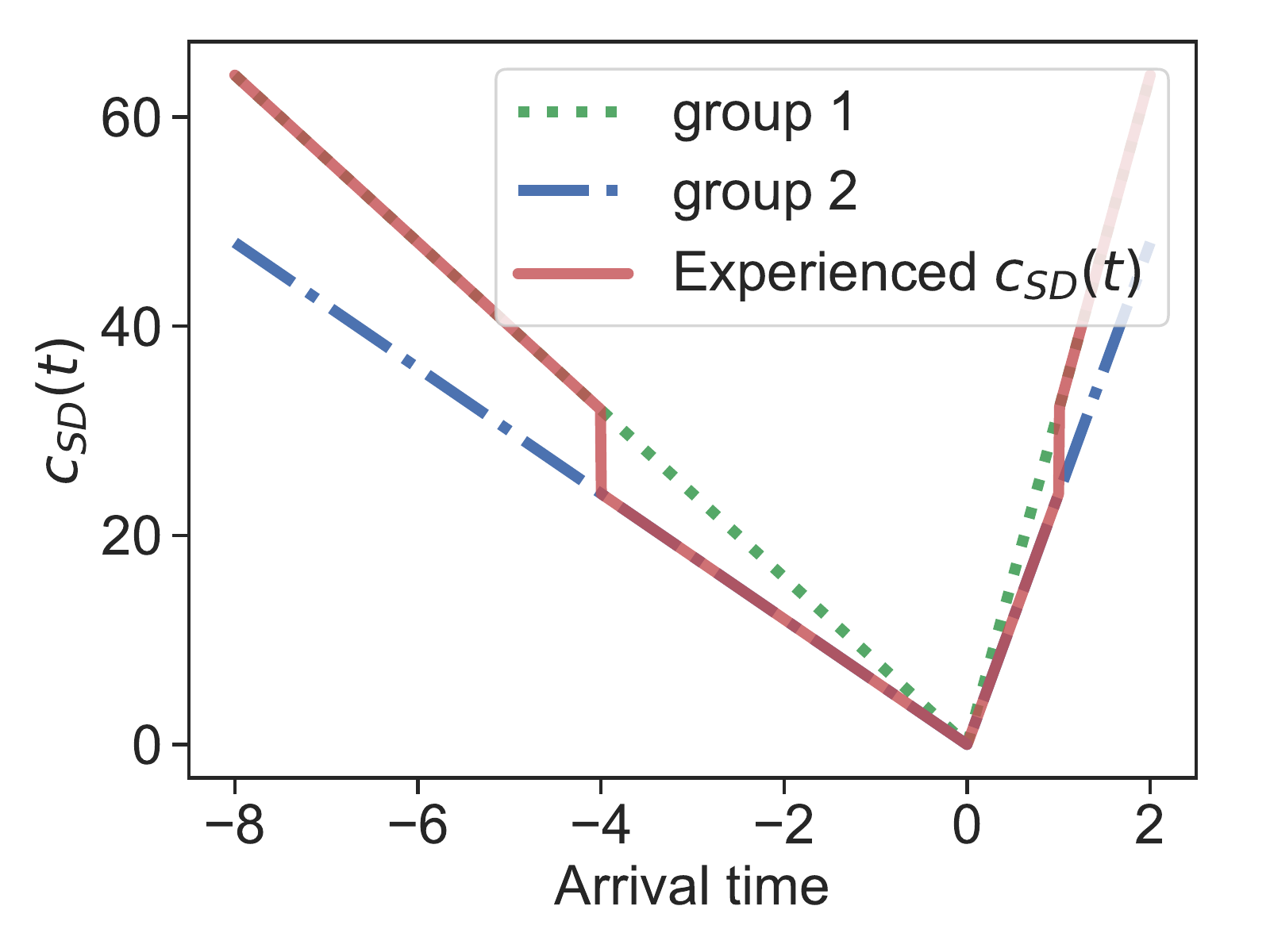}}\hfill
\subfloat[\label{fig:te2_toll}] {\includegraphics[width=0.33\textwidth]{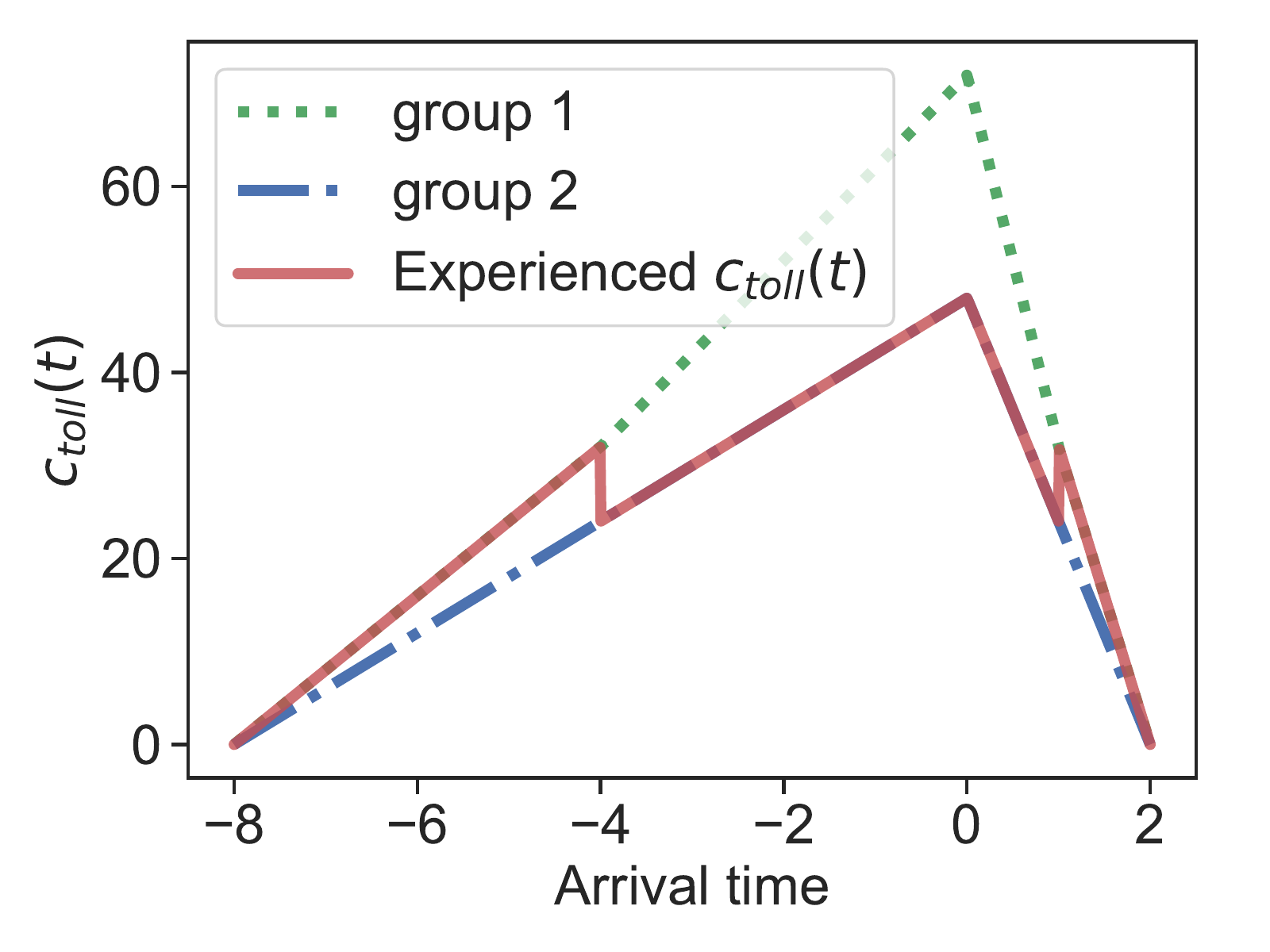}}\hfill
\subfloat[\label{fig:te2_total}]{\includegraphics[width=0.33\textwidth]{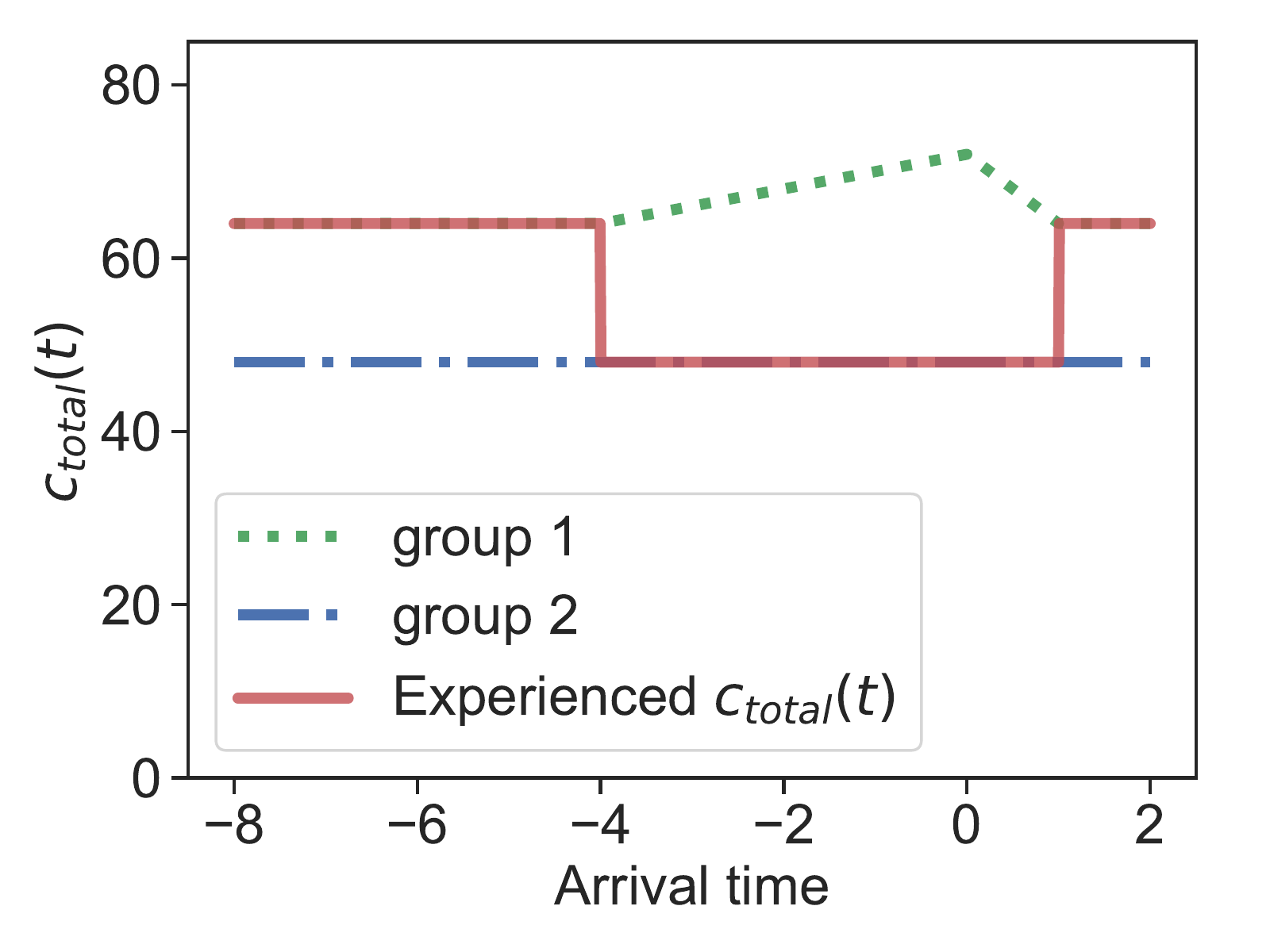}}\hfill
\caption{Plots of time-varying costs for the cases of no toll ((a),(b), and (c)), SO toll ((d),(e), and (f)), $\text{TE1}$ toll ((g),(h), and (i)), and $\text{TE2}$ toll ((j),(k), and (l)). For cases of no toll, TE1 toll, and TE2 toll group 1 travels at the boundaries of the peak period, while for the case of SO toll, the order is reversed.}
\label{fig:allTimeVaryingCostPlots}
\end{figure}

For the cases of SO, $\text{TE1}$, and $\text{TE2}$ tolls, the arrival time is same as the departure time. The plots \vpNew{in Figure 3} show the cost a traveler from group 1 or 2 would experience if they arrive the destination at a given time. The experienced cost (shown as an overlayed solid line) is the true experienced cost based on when the group travels. For the cases of no toll, $\text{TE1}$ toll, and $\text{TE2}$ toll, group 1 travels at the boundary of the peak period while group 2 travels in the middle. For the SO toll, the order is reversed. $\text{TE1}$ and $\text{TE2}$ tolls charge different toll for different groups. In contrast, \vpNew{the SO toll does not depend on which group is traveling during the time}; hence, the tolls in Figure~\ref{fig:so_toll} do not differ for each group. The toll plot for group 1 under $\text{TE2}$ toll in Figure~\ref{fig:te2_toll} considers the value of $\bar{\wp}=1.25$; however, the plot of experienced cost is same for any choice of $\bar{\wp}$ greater than 1.  Plots of $c_\text{toll}^k(t)$ for the no toll case and $c_\text{queue}^k(t)$ for the SO, $\text{TE1}$, and $\text{TE2}$ toll cases are not shown as their values are zero for both groups for the entire peak period.

The observations from Figure~\ref{fig:allTimeVaryingCostPlots} \vpNew{illustrate} our findings from previous sections. The $c_{\text{SD}}^k(t)$ costs for the $\text{TE1}$ and $\text{TE2}$ tolls are identical to the no-toll case since they preserve the order of departure. The $\text{TE1}$ toll exactly substitutes travel time cost for each group with the toll cost paid by the group (experienced cost plots in Figures~\ref{fig:no_queue} and~\ref{fig:te1_toll} are identical). In contrast, the $\text{TE2}$ toll increases the toll paid by group 2 to generate same revenue as the SO toll. This is apparent in the equal areas under the solid-line toll curves for Figures~\ref{fig:so_toll} and~\ref{fig:te2_toll}. Last, as expected, the total cost for each group is lowest during the time the group is traveling, since at equilibrium total costs are minimized for the chosen departure times.

Next, we quantify the benefit to the individuals and the community under different toll scenarios. Consider the schematic in Figure~\ref{fig:infograph} that shows the benefit and cost to the individuals and the community. We consider that tolls collected from individuals are not directly rebated but are invested back into the community.
\begin{figure}[H]
\centering
\includegraphics[scale=0.5]{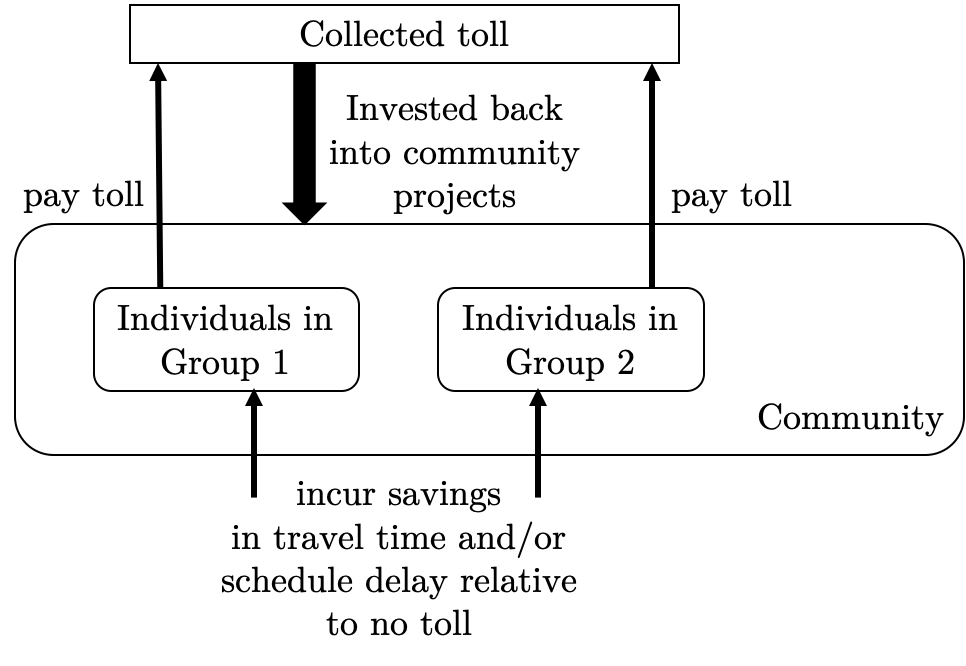}
\caption{Schematic showing the benefit and cost to individuals and community. Individuals' net benefit is the incurred savings relative to the paid toll. Community's net benefit is the total incurred savings and total revenue collected.}
\label{fig:infograph}
\end{figure}

 \paragraph{Benefit to the community:} Let us assume that the cost of installing the toll gantries (via open-road tolling methods) is same regardless of the type of tolling in place. Under this case, we can calculate the \textit{social benefit} to the community under a given toll scenario $x$, $x \in \{ \text{SO}, \text{TE1}, \text{TE2}\}$, as the sum of travel time savings for each individual relative to the no-toll case and the total revenue collected. That is,
 
\begin{align*}
	\text{Social benefit}_x &= \text{Savings relative to the no-toll case} + \text{Revenue collected} & \\
	&= \sum_{k \in K} \text{Savings}_k^x + TRC_k^{x} \qquad & \forall x \in \{ \text{SO}, \text{TE1}, \text{TE2}\}\\
	&= \sum_{k \in K} \left( SDC_k^{\text{NO}}+TTC_k^{\text{NO}}- SDC_k^{x}\right) + TRC_k^{x} \qquad & \forall x \in \{ \text{SO}, \text{TE1}, \text{TE2}\} \\
	&= \left( SDC^{\text{NO}}+TTC^{\text{NO}}- SDC^{x}\right) + TRC^{x} \qquad &\forall x \in \{ \text{SO}, \text{TE1}, \text{TE2}\}
\end{align*}
 where, $\text{Savings}^k_x = \left( SDC_k^{\text{NO}}+TTC_k^{\text{NO}}- SDC_k^{x}\right)$ is the total savings for  group $k$ under toll scenario $x$, $x \in \{ \text{SO}, \text{TE1}, \text{TE2}\}$, defined as the difference between the combined schedule delay and travel time costs for no toll case and the toll case $x$.

\paragraph{Benefit to an individual:} An individual pays toll in exchange for travel time delay and/or schedule delay savings relative to the no-toll case. Therefore, we define the net benefit received by an individual as follows: let $y_k^{x}$ be the ratio of social cost savings for a group $k$ relative to the no-toll case and the total toll paid by the group for a given toll scenario $x \in \{ \text{SO}, \text{TE1}, \text{TE2}\}$. That is,

\begin{align}
	y_k^{x} &= \frac{\text{Savings}_k^x}{TRC_k^{x}} \label{eq:yDef}\\
	&= \frac{SDC_k^{\text{NO}} + TTC_k^{\text{NO}} - SDC_k^{x} }{TRC_k^{x} } \qquad \forall x \in \{ \text{SO}, \text{TE1}, \text{TE2}\}, k \in K.
\end{align}

Intuitively, $y_k^{x}$ measures the benefit received by a group $k$ under the presence of toll regime $x$ relative to the toll paid by the group. A value of $y_k^{x}=1$ indicates that the group receives benefit in equal worth to the toll paid, which is an ideal scenario. A high difference between the values of $y_k^{x}$ for different groups is inequitable since a difference indicates that one group is worse off than the other. We define the difference between the maximum and minimum values of $y_k^{x}$ across different groups as the \textit{equity gap} for a given toll scenario. %
Next, we compare the social benefit and the equity gap for the three tolls.

Figure~\ref{fig:BaseCase_SocialCost} shows the variation of social benefit for all toll scenarios relative to the no-toll case, with the savings and toll paid split by each group. As expected, SO tolls generate the largest social benefit across all groups. This is attributed to the cost minimizing nature of SO tolls. While the savings across groups 1 and 2 are identical for $\text{TE1}$ and $\text{TE2}$ tolls, $\text{TE2}$ tolls generate $240\aleph$ higher social benefit that $\text{TE1}$ tolls. This is attributed to the revenue-neutral nature of $\text{TE2}$ tolls. While SO tolls generate the highest possible social benefit, the split of benefit across different groups is uneven leading to equity issues which we highlight next.

\begin{figure}[H]
	\centering
	\includegraphics[scale=0.5]{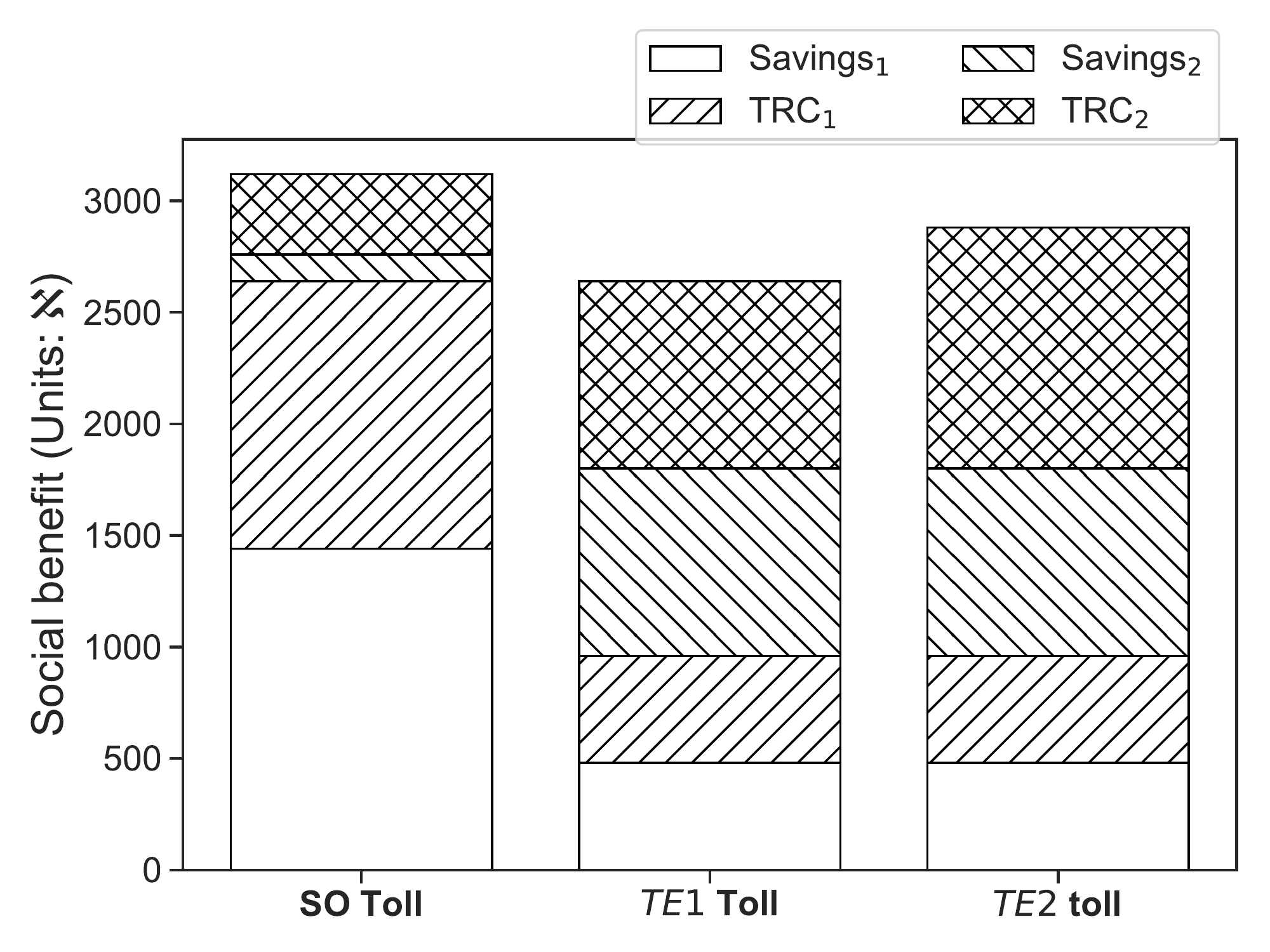}
	\caption{Social benefit across different toll scenarios for the base case scenario}
		\label{fig:BaseCase_SocialCost}
\end{figure}

 Figure~\ref{fig:costVsToll} shows the variation of $y_k^{x}$ for each group for the three toll scenarios. We observe that $\text{TE1}$ toll generates equitable benefit for both groups with  $y_1^{\text{TE1}}=y_2^{\text{TE2}}=1$. This is as expected: $\text{TE1}$ tolls preserve the order so the benefit received by a group is same as the group's travel-time cost under the no-toll scenario, which is identical to the toll paid by the group under the $\text{TE1}$ toll. As discussed in Section~\ref{subsec:SOtoll}, SO toll generates the highest differences between the two groups where group 1 with  $y_1^{\text{SO}}=1.2$ benefits more at the expense of group 2 with $y_2^{\text{SO}}=0.33$. $\text{TE2}$ tolls also generate differences between the group (because we make group 2 pay more than their savings to generate the same revenue as the SO toll), but the equity gap of $0.23$ (where, $y_1^{\text{TE2}}=1$ and $y_2^{\text{TE2}}=0.77$) is not as high as the SO toll.%

\begin{figure}[H]
	\centering
	\subfloat[\label{fig:y_SO}]{\includegraphics[width=0.33\textwidth]{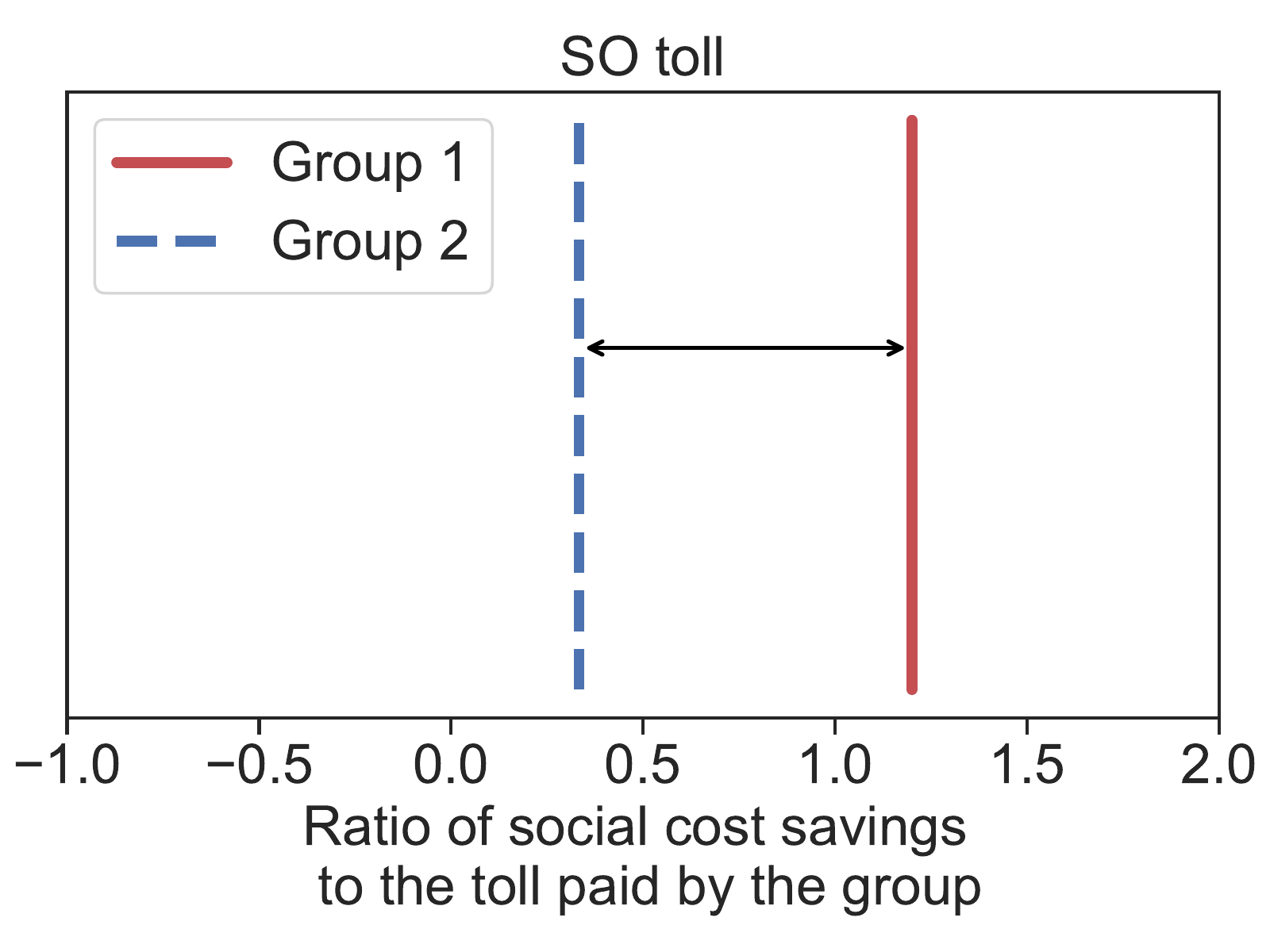}}\hfill
\subfloat[\label{fig:y_TE1}] {\includegraphics[width=0.33\textwidth]{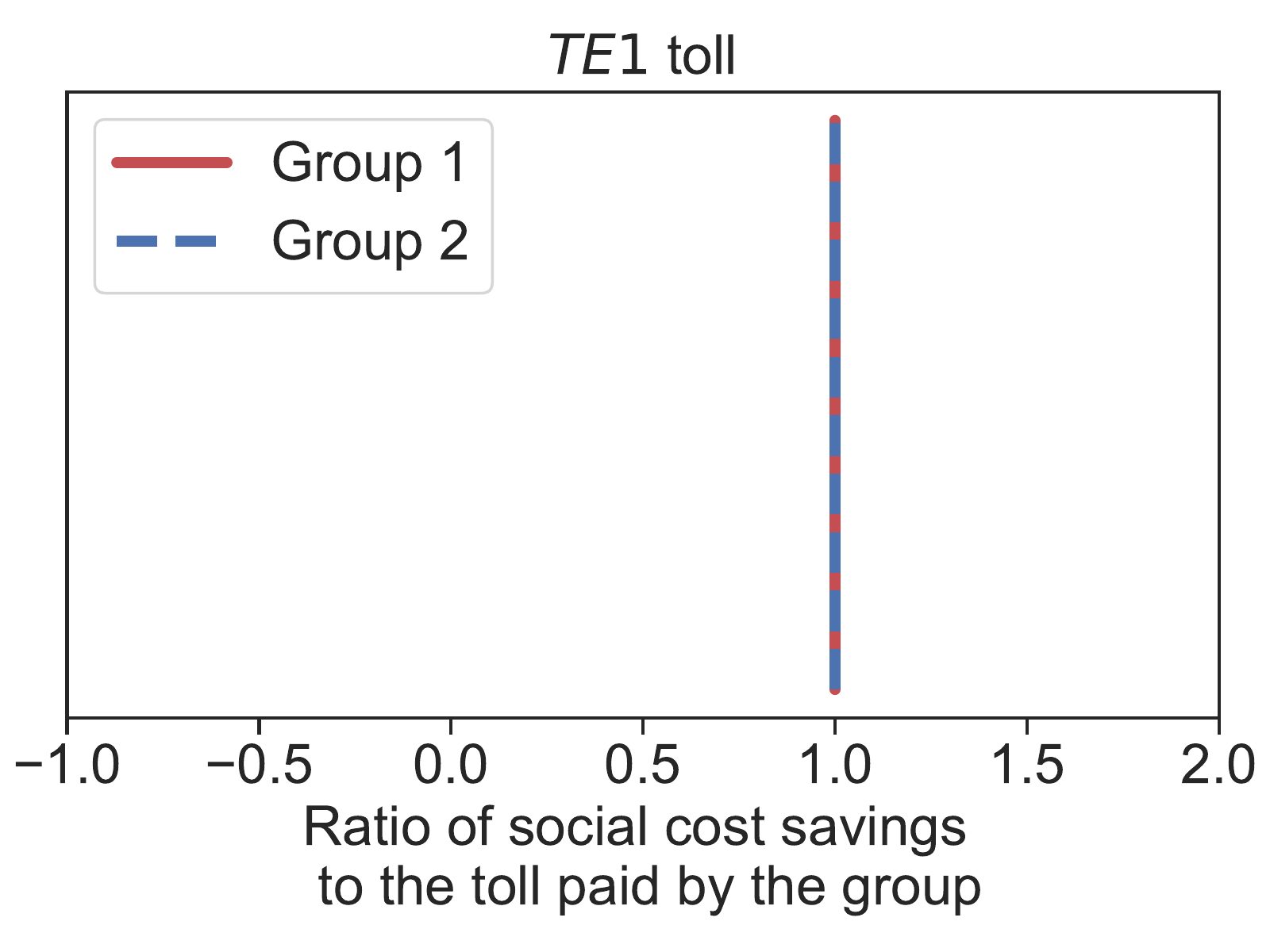}}\hfill
\subfloat[\label{fig:y_TE2}]{\includegraphics[width=0.33\textwidth]{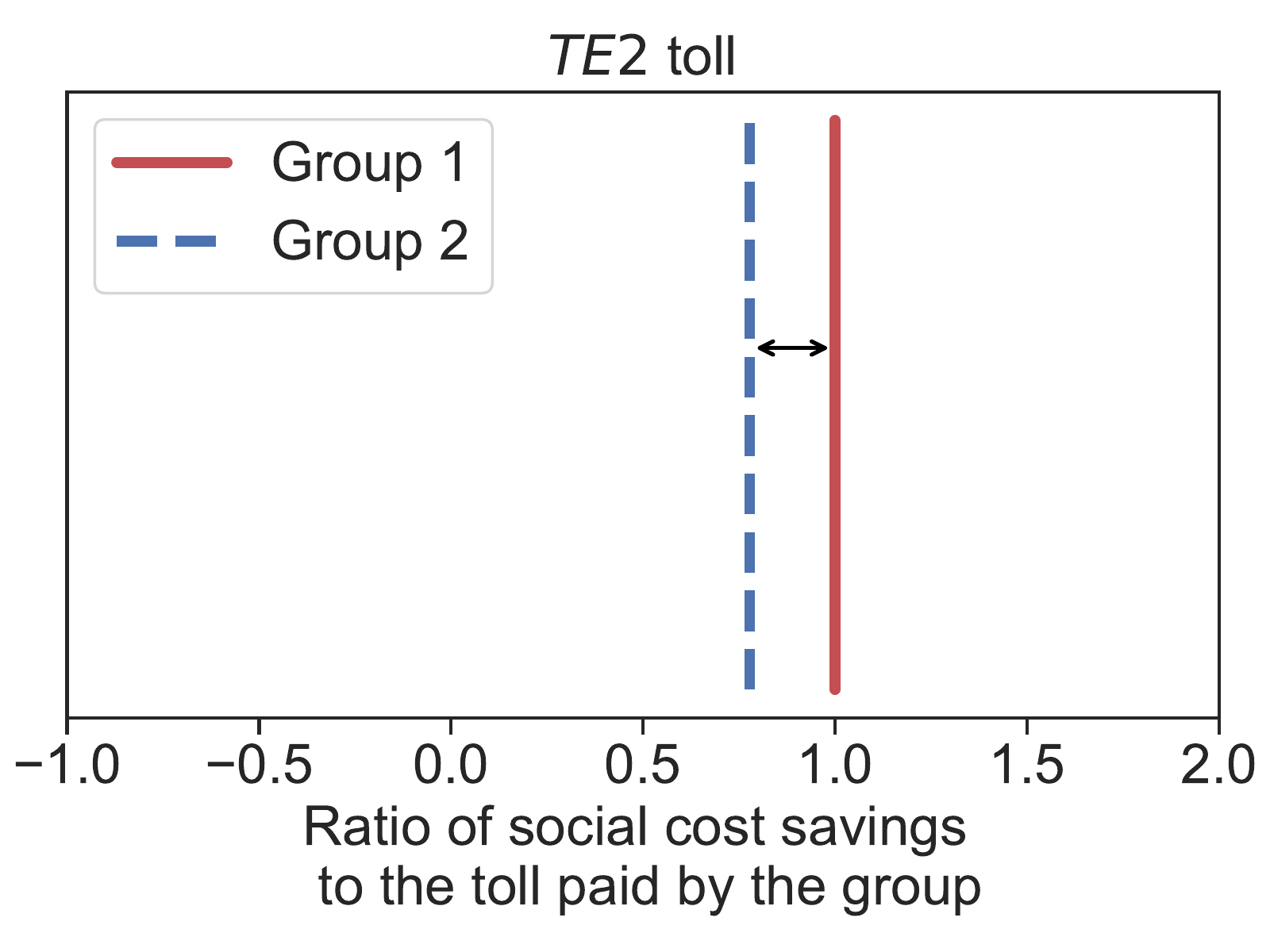}}\hfill
	\caption{Cost savings for each group relative to the toll paid. SO toll has the highest equity gap}
	\label{fig:costVsToll}
\end{figure}

Thus, we observe a tradeoff between the social benefit to the community and the equity gap between different groups. On one end, $\text{TE1}$ tolls generate a zero equity gap, but have the lowest social benefit, while on the other hand, SO tolls have the highest social benefit, but generate a high equity gap between the two groups. $\text{TE2}$ tolls provide a middle ground between the two extremes. Next, we conduct sensitivity analysis of social benefit and equity gap between the groups under different toll scenarios for varying values of parameter values.

\subsection{Sensitivity analysis}

In this section, we assume that the parameters of group 2 ($\beta_2$ and $\alpha_2$) and the number of travelers ($N$) are held constant as the base case. Variation in all other variables is considered. We only consider the case where $\beta_1>\beta_2$ so that the SO toll's equity issues are highlighted. Thus, $\beta_1/\beta_2$ is considered a variable with values higher than $1$. Similarly, given the criteria for original group numbering in Section 3, we consider $\alpha_1/\beta_1 > \alpha_2/\beta_2$. We conduct sensitivity of social benefit and equity gap against five variables: bottleneck capacity ($D$), ratio of late arrival penalty to early arrival penalty ($\eta$), proportion of group 2 travelers ($f_2$), the ratio $\beta_1/\beta_2$, and the ratio $(\alpha_1/\beta_1)/(\alpha_2/\beta_2)$. We vary one variable at a time keeping the other variables constant to their base values given by $D=6$ veh/min, $\eta=4$, $f_2=0.5$, $\beta_1/\beta_2 = 3/2$, and $(\alpha_1/\beta_1)/(\alpha_2/\beta_2) = 4/3$.

Figures~\ref{fig:sensD},~\ref{fig:sensEta},~\ref{fig:sensf},~\ref{fig:sensb1Byb2}, and~\ref{fig:sensa1b1Bya2b2} show the variation in equity gap and social benefit for varying values of $D$, $\eta$, $f_2$, $\beta_1/\beta_2$, and $(\alpha_1/\beta_1)/(\alpha_2/\beta_2)$, respectively. The equity gap is shown as a region where the upper boundary marks the value of $y_1^x$ for group 1, while the lower boundary of the region is the value of $y_2^x$ for group 2 for different toll scenarios $x \in \{ \text{SO}, \text{TE1}, \text{TE2} \}$. Although the equity gap region for the $\text{TE1}$ toll is a fixed line with a value of $1.0$, we still include this toll for completeness.

\begin{figure}[H]
	\centering
	\subfloat[\label{fig:equity_gap_D}]{\includegraphics[width=0.5\textwidth]{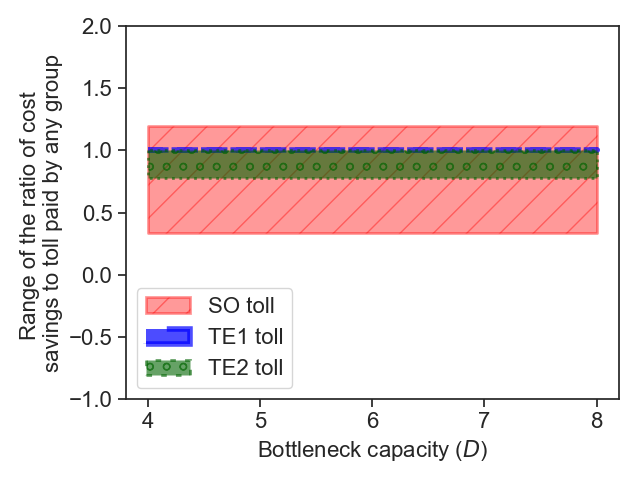}}\hfill
\subfloat[\label{fig:SC_saving_D}]{\includegraphics[width=0.5\textwidth]{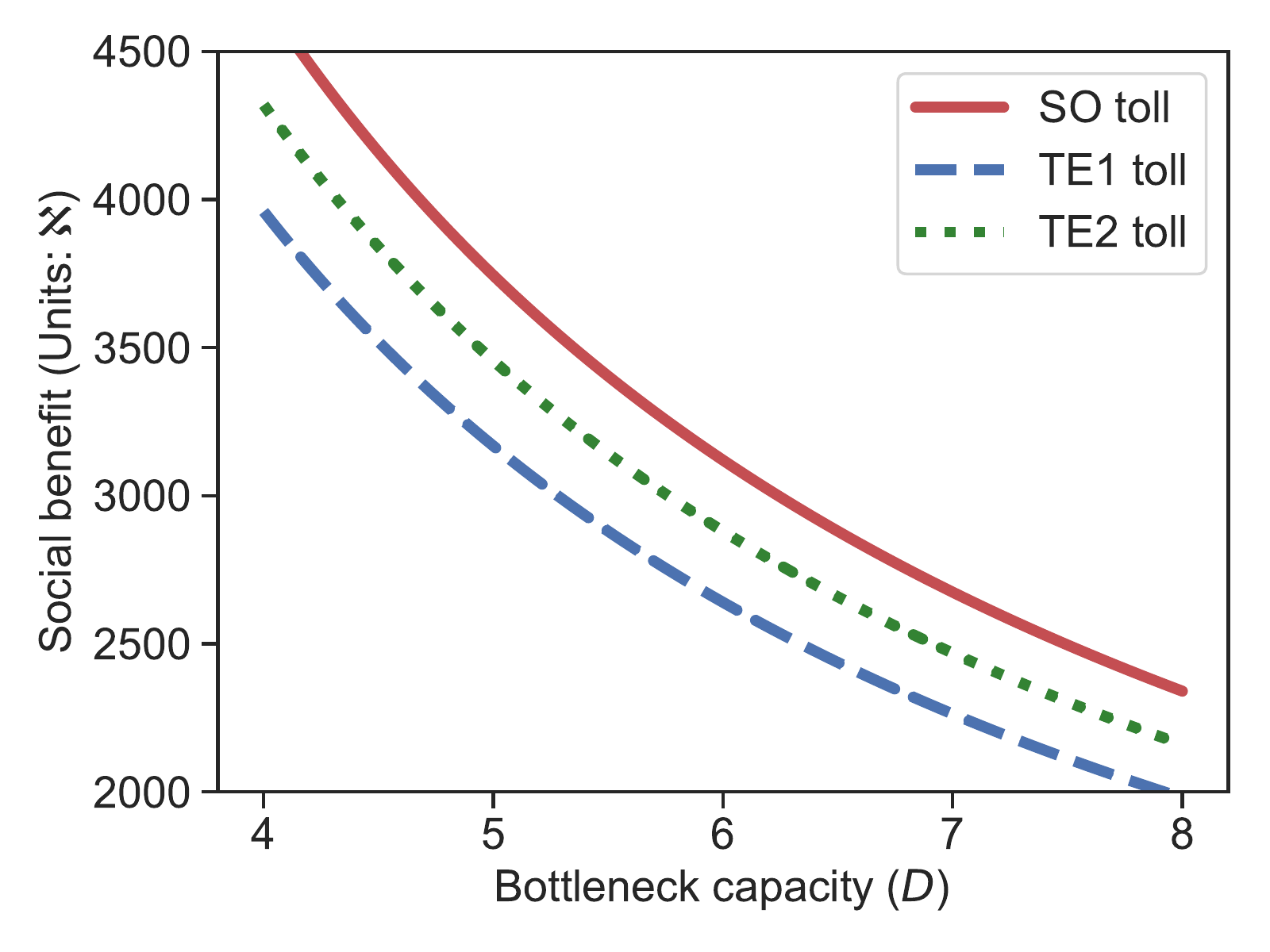}}\hfill
	\caption{Sensitivity of equity gap and social benefit relative to bottleneck capacity ($D$)}
	\label{fig:sensD}
\end{figure}

\begin{figure}[H]
	\centering
	\subfloat[\label{fig:equity_gap_eta}]{\includegraphics[width=0.5\textwidth]{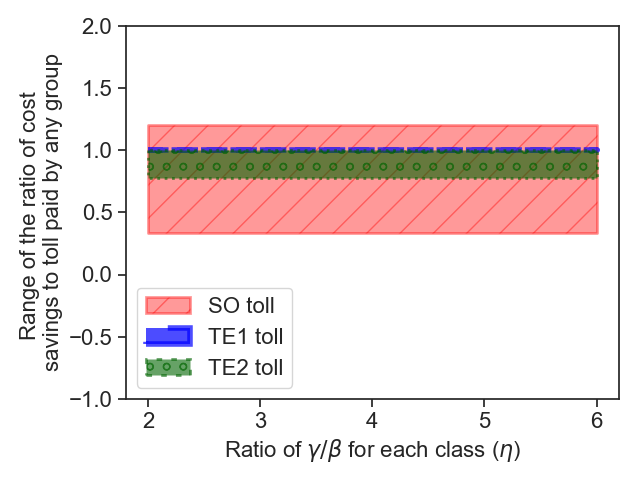}}\hfill
\subfloat[\label{fig:SC_saving_eta}]{\includegraphics[width=0.5\textwidth]{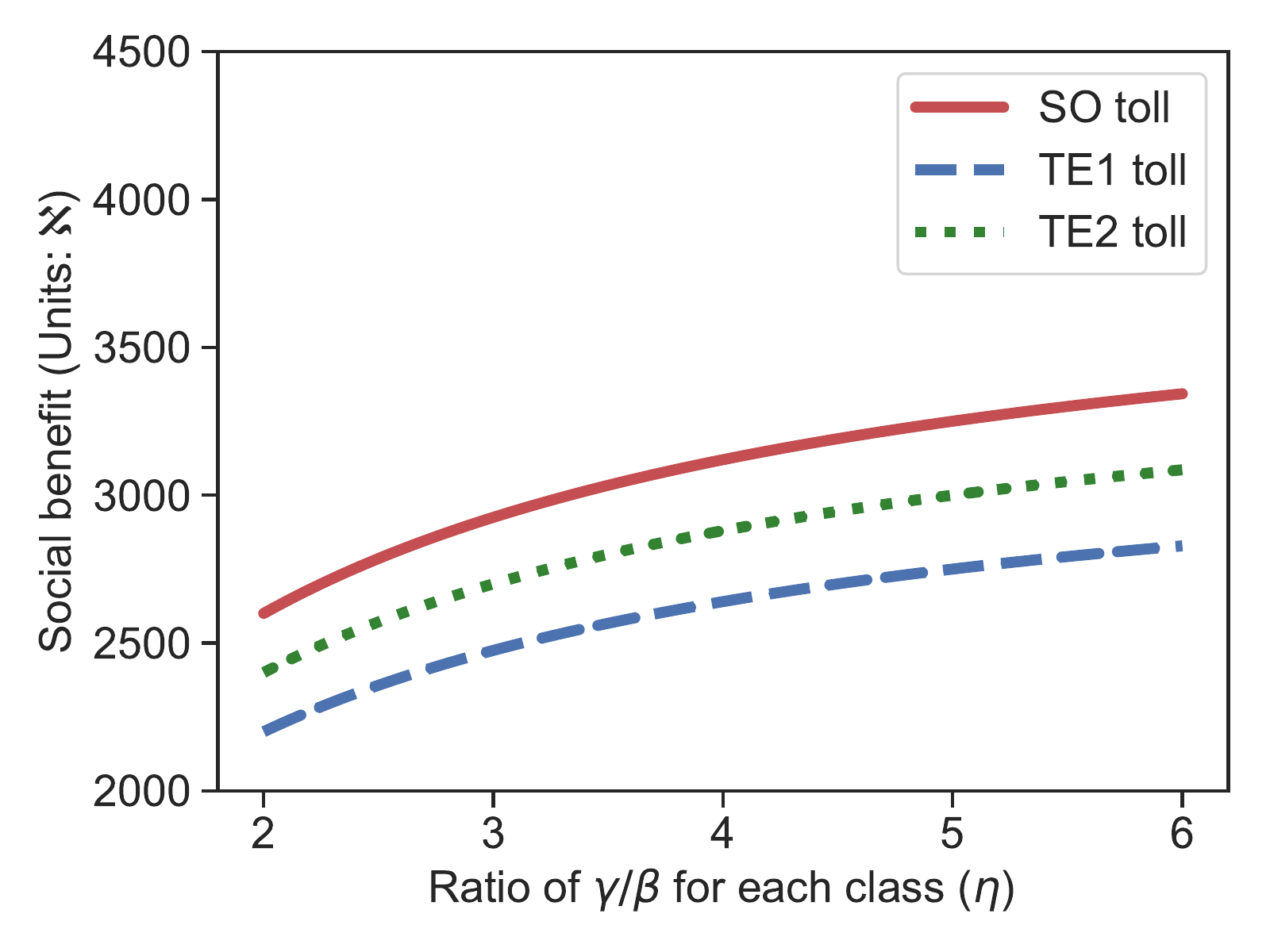}}\hfill
	\caption{Sensitivity of equity gap and social benefit relative to $\eta$}
	\label{fig:sensEta}
\end{figure}

\begin{figure}[H]
	\centering
	\subfloat[\label{fig:equity_gap_f2}]{\includegraphics[width=0.5\textwidth]{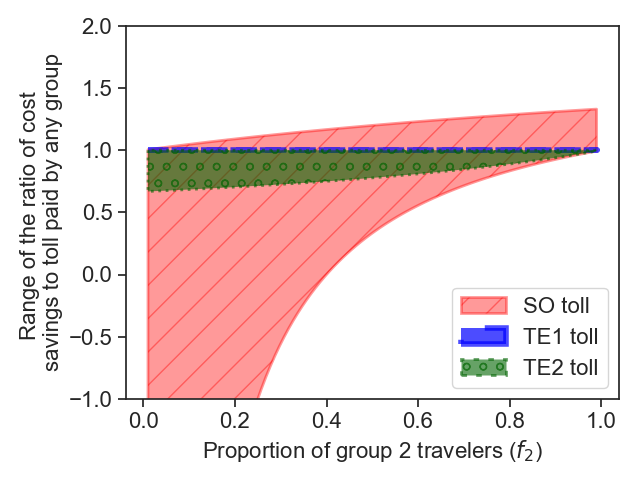}}\hfill
	\subfloat[\label{fig:SC_saving_f2}]{\includegraphics[width=0.5\textwidth]{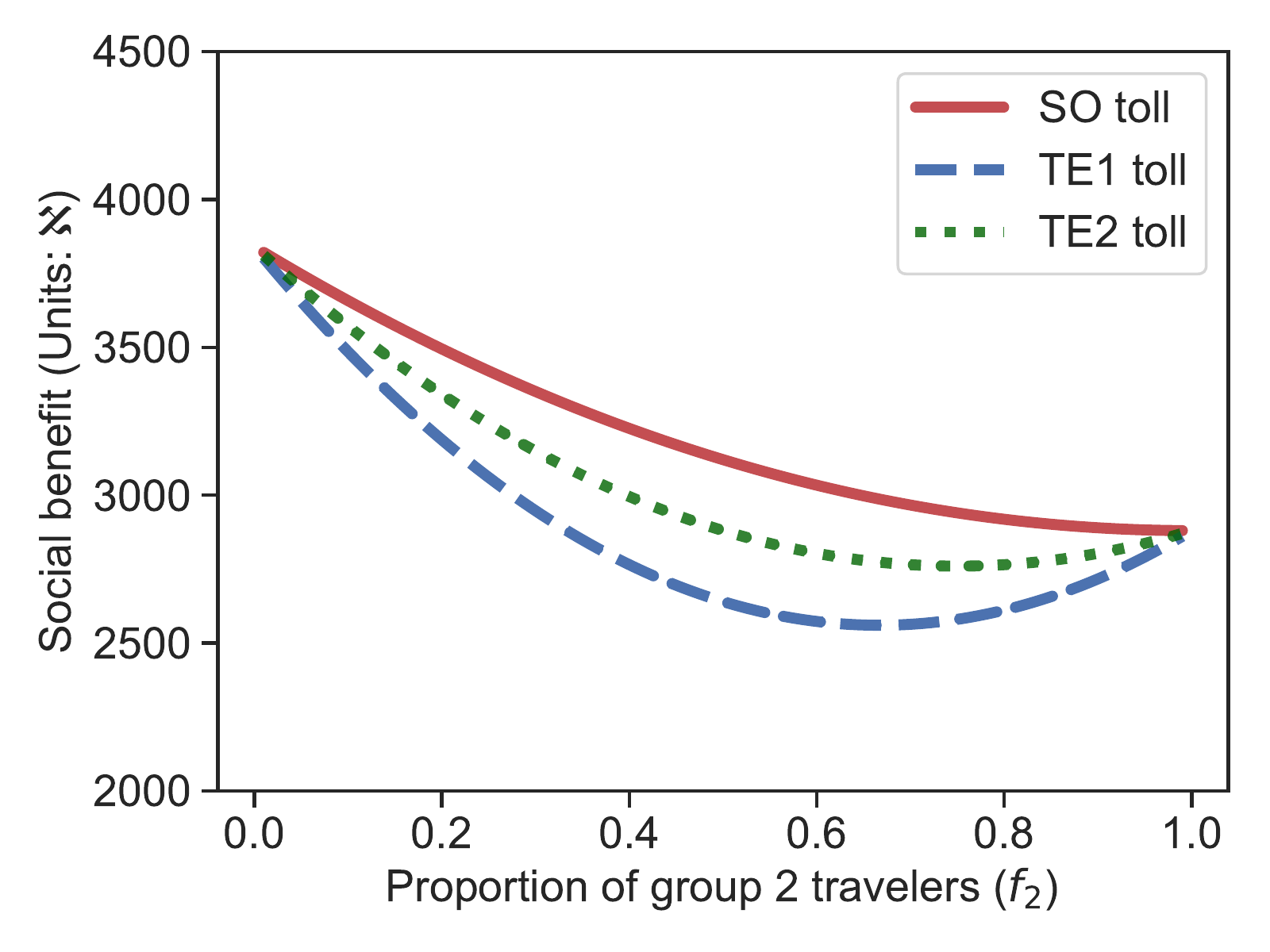}}\hfill
	\caption{Sensitivity of equity gap and social benefit relative to proportion of group 2 travelers $f_2$}
	\label{fig:sensf}
\end{figure}

\begin{figure}[H]
	\centering
	\subfloat[\label{fig:equity_gap_beta}] {\includegraphics[width=0.5\textwidth]{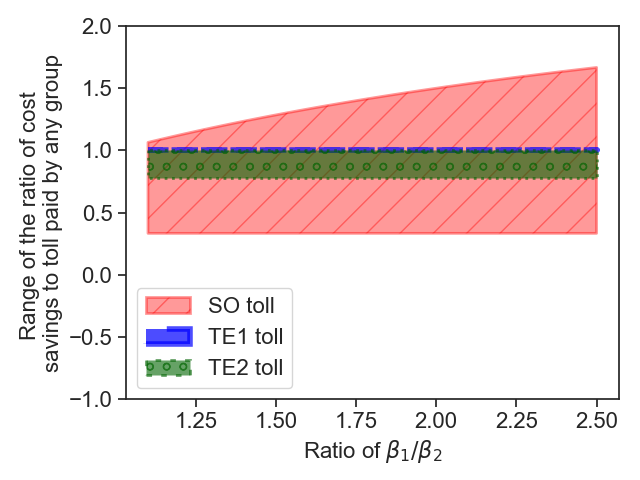}}\hfill
\subfloat[\label{fig:SC_saving_beta}] {\includegraphics[width=0.5\textwidth]{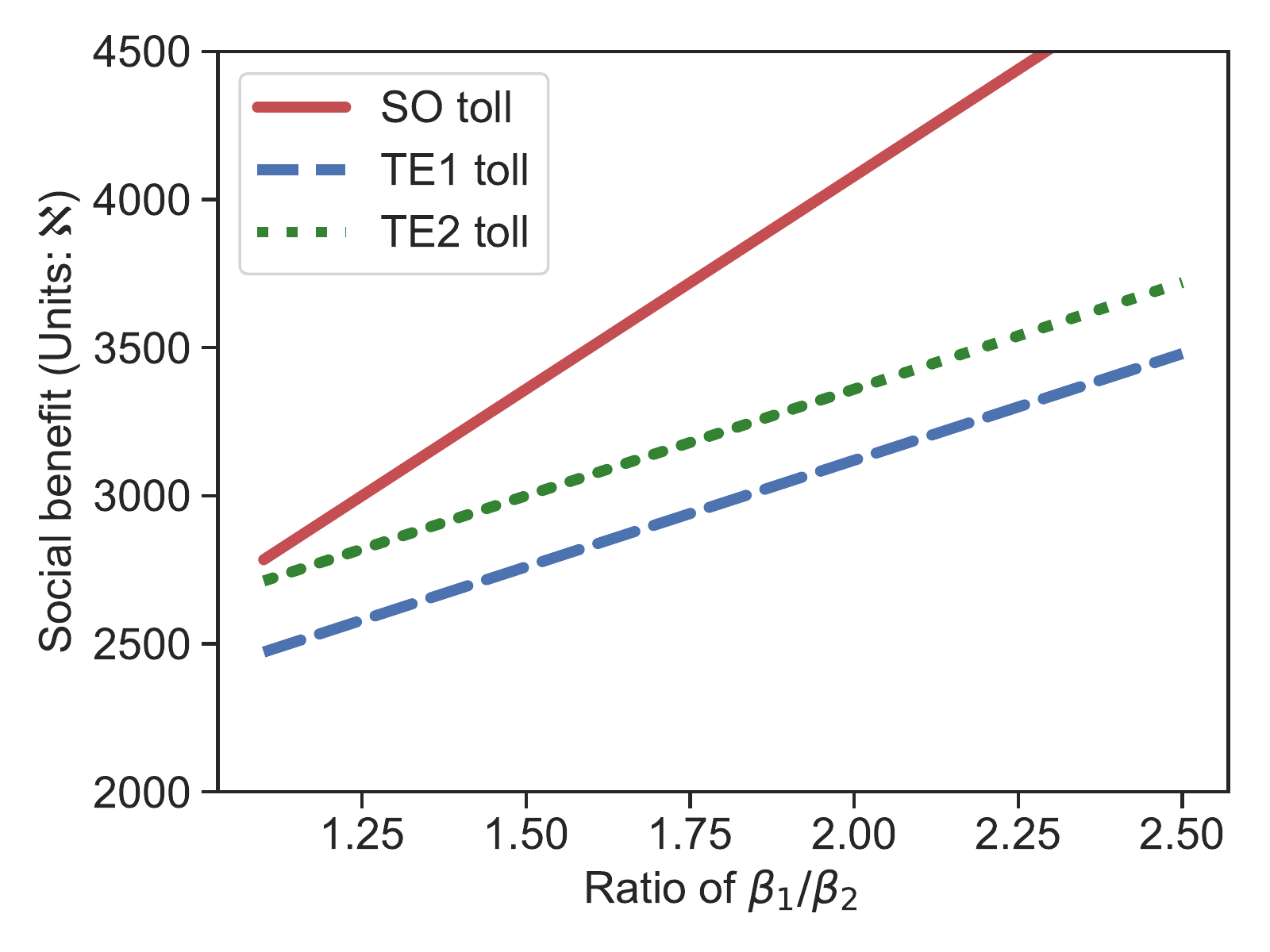}}\hfill
	\caption{Sensitivity of equity gap and social benefit relative to the ratio of  $\beta_1/\beta_2$}
	\label{fig:sensb1Byb2}
\end{figure}

\begin{figure}[H]
	\centering
\subfloat[\label{fig:equity_gap_alpha}]{\includegraphics[width=0.5\textwidth]{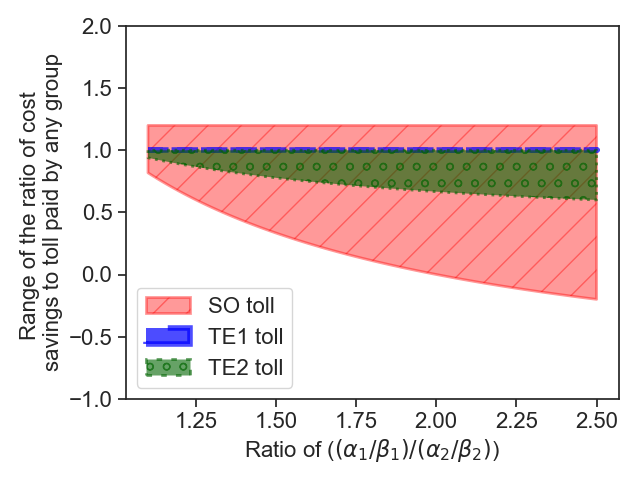}}\hfill
\subfloat[\label{fig:SC_saving_alpha}]{\includegraphics[width=0.5\textwidth]{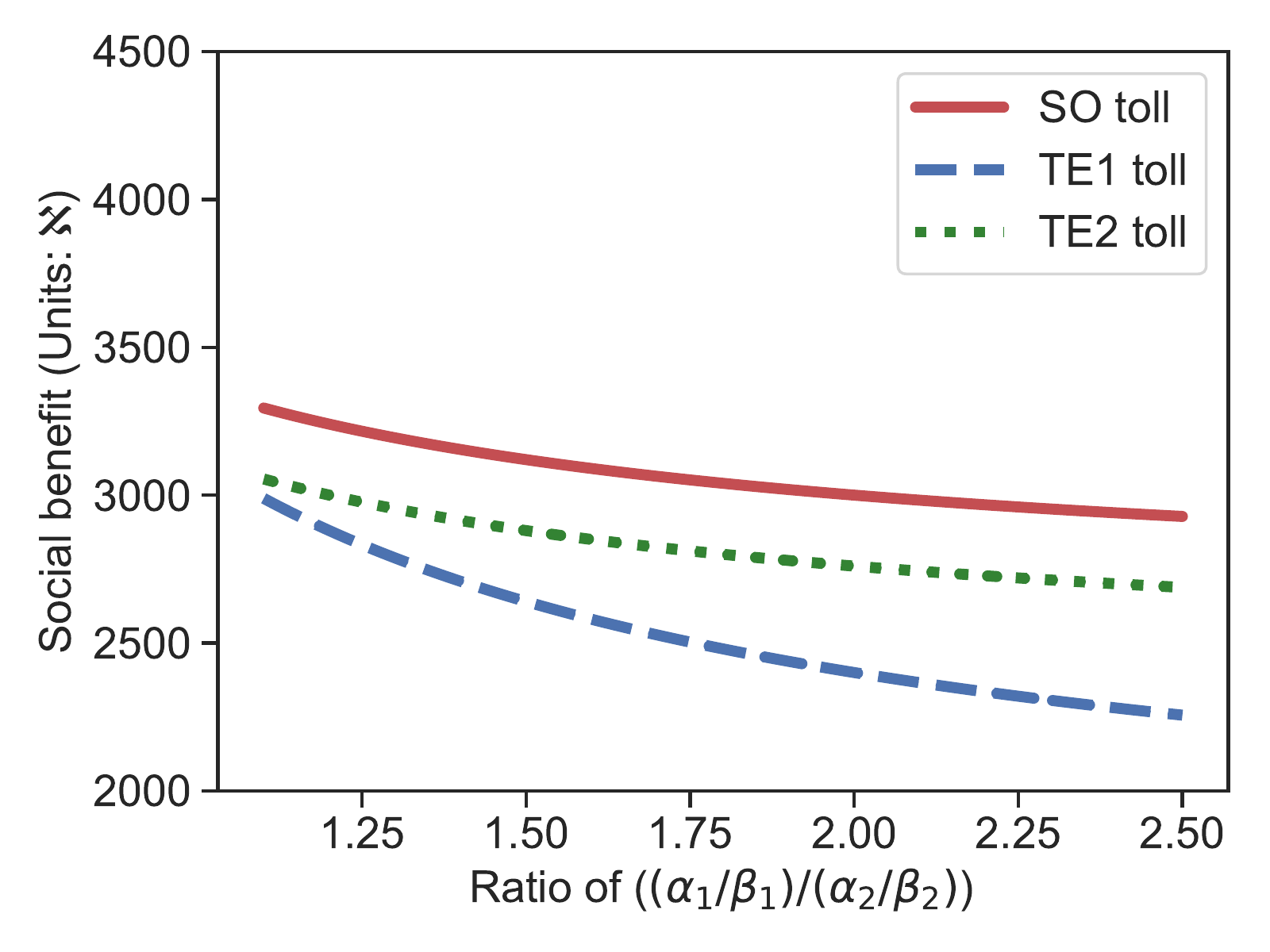}}\hfill
	\caption{Sensitivity of equity gap and social benefit relative to the ratio of  $(\alpha_1/\beta_1)/(\alpha_2/\beta_2)$}
	\label{fig:sensa1b1Bya2b2}
\end{figure}

We make \vpNew{the} following observations:
\begin{enumerate}
	\item Similar to the base case scenario, the social benefit and the equity gap for the SO tolls are always higher than the $\text{TE1}$ and $\text{TE2}$ tolls.
	\item Figure~\ref{fig:SC_saving_D} shows that as the bottleneck capacity increases, social benefit decreases. This is because the congestion under the no-toll case is lower when more capacity is available to meet the same demand, resulting in lower savings after charging a toll. Similarly, in Figure~\ref{fig:SC_saving_eta},  social benefit increases with an increase in the value of $\eta$, since cost savings are a function of the term $\eta/(1+\eta)$ which increases as $\eta$ increases. Furthermore, the equity gap for the SO, $\text{TE1}$ and $\text{TE2}$ tolls remain constant for all values of $D$ and $\eta$ (keeping other variables constant). This can be seen after substituting the cost expressions in Equation~\eqref{eq:yDef} where the $\eta$ and $D$ terms in the numerator and the denominator cancel out.
	\item  Figure~\ref{fig:SC_saving_f2} shows the variation in social benefit for different values of $f_2$. At the extreme values ($f_2=0$ and $f_2=1$), the population comprises of only one group. This results in identical social cost savings under different toll scenarios, which is because for each toll scenario, the toll increases (or decreases) at the rate of $\beta$ (or $\gamma$) in the time period when the group travels. For the intermediate values of $f_2$, the social benefit achieves a minima for the $\text{TE1}$ and $\text{TE2}$ tolls where the point of minima depends on the ratios of $(\alpha_1/\beta_1)/(\alpha_2/\beta_2)$ and $\beta_1/\beta_2$ for each group. Since we are interested in maximizing the social benefit, the largest benefit is obtained only if all travelers in the population are uniform and have high $\alpha$ and $\beta$  values.
	\item Figure~\ref{fig:equity_gap_f2} shows that as the value of $f_2$ decreases the equity gap between the two groups increases for both $\text{TE2}$ and SO tolls. The value of the ratio evaluated from Equation~\eqref{eq:yDef} may be negative if the value of $\text{Savings}_k^x$ for a group is negative (that is, the group has higher schedule delay under the presence of a toll than the combined travel time delay  and schedule delay under the no-toll case.) The equity gap is undefined when $f_2=0$ or when $f_2=1$ as there is only one group under those settings. For the case when $f_2$ is very close to zero, $y_2^{\text{SO}}\rightarrow -\infty$ and the equity gap for the SO toll case tends to $\infty$. This is as expected: under SO tolls, group 2 is pushed to travel at the boundary of the peak period; if there are relatively few travelers from group 2, they face much higher schedule delay, while continuing to pay tolls, and the rate of increase of tolls is higher than the rate of decrease of savings.  This finding shows us that if the population for travelers who are absolutely more time-flexible travelers but are relatively less time-flexible is really small, they suffer significantly in order to provide positive savings for the other travel group. A large equity gap is attributed to the reversal of the order of travel which contributes to the time-inequity of SO tolls. The equity gap generated by the $\text{TE2}$ tolls is also higher for lower values of $f_2$, but not as high as the SO tolls.
	\item Increasing the ratio of $\beta_1/\beta_2$ increases the equity gap for the SO toll scenario, but the equity gap under TE2 and TE1 tolls remain the same. Additionally, the social benefit increases with increasing value of $\beta_1/\beta_2$ for all tolls, but the rate of increase is higher for the SO tolls. This pattern exists because the savings of each group under $\text{TE1}$ and $\text{TE2}$ tolls depend on their relative willingness to pay, not their absolute values. On the other hand, the SO toll depends on the absolute values of $\beta_1$ and $\beta_2$ and thus increasing the ratio increases the social benefit (at a higher rate than $\text{TE1}$ and $\text{TE2}$ tolls) as well as the equity gap. 
	
	\item Last, we observe that increasing the value of the $(\alpha_1/\beta_1)/(\alpha_2/\beta_2)$ increases the equity gap but decreases the social benefit. This is as expected: as the relative time-flexibility of group $2$ decreases relative to group $1$, group $2$ travelers are more worse off for both SO and $\text{TE2}$ tolls. %

\end{enumerate}

The sensitivity analysis helps us identify that the equity gap for the SO toll is highly sensitive to the relative population of two groups and the values of parameters determining absolute and relative time-flexibility.

\section{Conclusion}
\label{sec:conclusion}
In this article we discussed equity issues for dynamic tolls. Building on the idea that ``all social goods are to be distributed equally, unless an unequal distribution would be to everyone's advantage," we made a case for the inclusion of time poverty  through the use of vertically-equitable tolls for dynamic pricing. We showed that the commonly-used SO tolls present two distinct equity concerns. First, they achieve their efficient outcome by shifting a greater percentage of system costs to low income travelers. Second, SO tolls are likely to increase time poverty among low income travelers since these tolls order traveler departures on the absolute value of traveler time rather than the relative value of schedule delay to travel time delay costs that orders no-toll equilibria.  The result is that structuring traveler departures so that low income travelers depart at the margins of the peak hour means that already overburdened families face significant burdens of time poverty.

As an alternative to the inefficiencies of the SO toll, we proposed time-equitable tolls that preserve the order of departure  under no toll, create zero queue, and generate same revenue as the SO tolls. Using numerical experiments for two groups of travelers, we showed that while SO tolls generate the highest possible social benefit across all tolling scenarios, they lead to a wide equity gap between different groups of travelers. In contrast, time-equitable tolls sacrifice some of the social benefit in exchange for lowering the equity gap in the population. These insights into the trade offs between system-optimal efficiency and equity will help prioritize tolling objectives for future dynamic tolling schemes and inform empirical research into how those objectives may be fulfilled.

\subsection{Implementation concerns about time-equitable toll}
\label{subsec:TE2tollOtherConcerns}
Implementing vertically-equitable tolls that vary across different groups in the real world brings challenges that are worth discussing. There are \vpNew{four main concerns about implementing these tolls}:
\begin{enumerate}[(a), topsep=0pt]
	\item The proposed tolls are time-dependent and vary continuously with time. While a continuous toll variation is possible (\vpNew{we can record the exact time when the vehicle entered the bottleneck and charge an appropriate toll}), it is desirable that tolls are incremented in discrete time intervals for the ease of public understanding. A few alternatives have been presented to fix this problem. These include replacing dynamic tolls with a flat peak-period toll~\citep[also called coarse tolls, see][]{arnott1990economics,xiao2011morning}, or charging step tolls that increment toll values in discrete steps of time~\citep{lindsey2010step}. Both coarse tolls and step tolls are easy to implement as they do not vary continuously with time; however, it is unclear whether they can be customized for achieving time equity. A detailed mathematical analysis of how these tolls will impact the equilibrium departure patterns and the equity gap is left as part of the future work.
	\item \vpNew{The tolls are not bounded above.} In our analysis, the highest toll values  depend on the values of parameters such as $\alpha, \beta$, and $\gamma$ values of a group and the number of travelers and the capacity of the bottleneck. Because the highest toll is inversely proportion to the bottleneck capacity, an agency can lower the highest toll charge by providing more capacity. Alternatively, our analysis assumes inelastic demand regardless of how high the toll is. If tolls are high, that will deter the travelers from choosing the tolled facility, thereby lowering the demand and thus the toll. We leave a detailed analysis of the highest possible toll under elastic demand for future work.
	\item \vpNew{The values of $\alpha$, $\beta$, and $\gamma$ must be known.}  The analysis relies on good-enough estimates of these parameters for determining the tolls and for analyzing the social benefit and equity gaps. In our opinion, our choice of discrete number of groups allows more flexibility. Using a survey conducted on a sample of travelers who use a freeway during peak periods, we can estimate the values of $\alpha$, $\beta$, and $\gamma$ and combine them into as many number of groups as the analysis desires. The consequence of loosely estimated values of $\alpha$, $\beta$ or $\gamma$ is that the true experience of benefits and equity gap might be skewed. In such cases, we recommend the use of post-toll-implementation surveys that can help us measure the true benefits received by a traveler of each group. %
	\item \vpNew{The time-equitable toll requires different tolls for distinct groups.} There are a few ways to address this concern. One alternative is to code driver characteristics as part of the vehicle (like a toll tag that is encoded with demographic and time-preference parameters). This allows a vehicle to be categorized into one group or the other, which can then be used to charge different tolls. However, such use of vehicle tags can incentivize fraudulent activities such as relatively time-rich travelers selling their tags to relatively time-poor travelers. The other alternative is to propose a rebate scheme where everyone pays a uniform toll, but travelers are then rebated similar to income-tax return based on their usage of the facility. This is similar to the idea of tradable bottleneck permit in \cite{akamatsu2017tradable}; however we leave a detailed analysis as part of the future work. %
\end{enumerate}

We hope that future researchers will investigate ways in which these implementation challenges can be addressed.

\subsection{Future work}
There are several other directions for future work.  First, an analysis using \vpNew{more realistic congestion} is warranted. Introducing networks with alternate paths, networks with multiple destinations, bottlenecks with stochastic capacity, travelers who desire to arrive throughout the peak hour, and travelers who can choose not travel to \vpNew{the bottleneck} will provide \vpNew{further insights on} tolls and equity. %
Second, future research should investigate ways for empirical estimations of time poverty parameters and \vpNew{place them in a context that policy makers will understand and appreciate.} %
 Last, the future research should address real world implementations of dynamic tolls with the equity analysis proposed in this thesis. The real world is messy. Real cities have hundreds of thousands of travelers who all have hundreds of variables, most of which are unobservable, affecting their decision to travel and this greatly complicates the problem of determining how tolls should be priced in order to equitably minimize congestion. An important piece of this puzzle, however, is to dive into the details of actual traveler decisions. This  remains an important next step for future research.

\section*{Acknowledgment}
Partial support for this research was provided by the Data-Supported Transportation Operations and Planning University Transportation Center, and the National Science Foundation Grants No. $1254921$, $1562291$, and $1826230$. The authors are grateful for this support.

\bibliography{VP_allReferences}

\appendix
\section{Cost expressions for the three tolling scenarios}
\label{appendix:costs}
We write the total cost expressions for different toll scenarios here, assuming two travel groups. In addition to the notation defined earlier, define $f_2 = N_2/N$ as the proportion of travelers belonging to group 2.

\subsection{No-toll equilibrium}
\textbf{Schedule delay costs}:
\begin{align}
	SDC_1 &= \frac{\beta_1 \eta}{2(1+\eta)}  \frac{N^2}{D} (1- f_2^2) \label{eq:SDC_noTollBegin}\\
	SDC_2 &= \frac{\beta_2 \eta}{2(1+\eta)}  \frac{N_2^2}{D} = \frac{\beta_1 \eta}{2(1+\eta)}  \frac{N^2}{D} \frac{\beta_2}{\beta_1} f_2^2 \\
	SDC &= \frac{\beta_1 \eta}{2(1+\eta)}  \frac{N^2}{D} \left[ 1+ \left( \frac{\beta_2}{\beta_1} -1 \right) f_2^2 \right] \label{eq:SDC_noTollEnd}
\end{align}

\textbf{Travel time costs}
\begin{align}
	TTC_1 &= \frac{\beta_1 \eta}{2(1+\eta)}  \frac{N^2}{D} (1- f_2)^2 \label{eq:TTC_noTollBegin} \\
	TTC_2 &= \frac{\beta_1 \eta}{2(1+\eta)}  \frac{N^2}{D}  f_2 \left[ 2\frac{\alpha_2}{\alpha_1} + \left( \frac{\beta_2}{\beta_1} - 2 \frac{\alpha_2}{\alpha_1} \right) f_2 \right] \\
	TTC &= \frac{\beta_1 \eta}{2(1+\eta)}  \frac{N^2}{D}  \left[ 1- 2 \left( 1- \frac{\alpha_2}{\alpha_1} \right)f_2 + \left( 1+ \frac{\beta_2}{\beta_1} -2 \frac{\alpha_2}{\alpha_1}\right) f_2^2 \right] \label{eq:TTC_noTollEnd}
\end{align}

\textbf{Toll costs}
Under no-toll, toll revenue is zero. That is, $TRC_1=TRC_2=TRC=0$.

\textbf{Total costs}
\begin{align}
	TC_1 &= \frac{\beta_1 \eta}{(1+\eta)}  \frac{N^2}{D} (1- f_2) \label{eq:TC_noTollBegin}\\
	TC_2 &= \frac{\beta_1 \eta}{(1+\eta)}  \frac{N^2}{D}  f_2 \left[ \frac{\alpha_2}{\alpha_1} +  \left( \frac{\beta_2}{\beta_1} - \frac{\alpha_2}{\alpha_1}  \right) f_2 \right] \\
	TC &= \frac{\beta_1 \eta}{(1+\eta)}  \frac{N^2}{D}  \left[ 1 + \left( \frac{\alpha_2}{\alpha_1} - 1 \right)f_2 + \left( \frac{\beta_2}{\beta_1} - \frac{\alpha_2}{\alpha_1} \right) f_2^2 \right] \label{eq:TC_noTollEnd}
\end{align}

\subsection{Tolled equilibrium: order preserved $\beta_2 > \beta_1$}
In this case, the order of travel is preserved and the SO, TE1, and TE2 tolls are identical. %

\textbf{Schedule delay costs}
The schedule delay costs are same as the no-toll case (Equation~\eqref{eq:SDC_noTollBegin}--\eqref{eq:SDC_noTollEnd}) because the order is preserved and the same proportion of each group travels early.

\textbf{Travel time costs}
Since there is no queue, travel time costs are zero under the modeling choice that sets free-flow travel time $(c_\text{FF}^k)(t)$ is set as zero for all groups $k\in\{1,2\}$. That is, $TTC_1=TTC_2=TTC=0$

\textbf{Toll costs}
\begin{align}
	TRC_1 &= \frac{\beta_1 \eta}{2(1+\eta)}  \frac{N^2}{D} (1- f_2)^2 \label{eq:TRC_TollBeginOrderSame} \\
	TRC_2 &= \frac{\beta_1 \eta}{2(1+\eta)}  \frac{N^2}{D}  f_2 \left[ 2(1-f_2) +  \frac{\beta_2}{\beta_1}  f_2 \right] \\
	TRC &= \frac{\beta_1 \eta}{2(1+\eta)}  \frac{N^2}{D}  \left[ 1+ f_2^2 \left( \frac{\beta_2}{\beta_1}-1 \right)  \right] \label{eq:TRC_TollEndOrderSame}
\end{align}

\textbf{Total costs}
\begin{align}
	TC_1 &= \frac{\beta_1 \eta}{(1+\eta)}  \frac{N^2}{D} (1- f_2) \label{eq:TC_TollOrderSameBegin}\\
	TC_2 &= \frac{\beta_1 \eta}{(1+\eta)}  \frac{N^2}{D}  f_2 \left[ 1+  \left( \frac{\beta_2}{\beta_1} - 1  \right) f_2 \right] \\
	TC &= \frac{\beta_1 \eta}{(1+\eta)}  \frac{N^2}{D}  \left[ 1 + \left( \frac{\beta_2}{\beta_1} - 1 \right) f_2^2 \right] \label{eq:TC_TollOrderSameEnd}
\end{align}
As observed, total cost for group 1 stays the same as the no-toll case. The total cost for group 2 is higher under no-toll case only if $\alpha_2 > \alpha_1$.

\subsection{Tolled Equilibrium: $\beta_1 > \beta_2$}

\subsubsection{SO toll}
For the SO toll, the order is reversed.

\textbf{Schedule delay costs} 
\begin{align}
	SDC_1 &= \frac{\beta_1 \eta}{2(1+\eta)}  \frac{N_1^2}{D} = \frac{\beta_2 \eta}{2(1+\eta)}  \frac{N^2}{D} \frac{\beta_1}{\beta_2} (1-f_2)^2 \label{eq:SDC_TollReverseBegin}\\
	SDC_2 &= \frac{\beta_2 \eta}{2(1+\eta)}  \frac{N^2}{D} (1- (1-f_2)^2) \\
	SDC &= \frac{\beta_2 \eta}{2(1+\eta)}  \frac{N^2}{D} \left[ 1+ \left( \frac{\beta_1}{\beta_2} -1 \right) (1-f_2)^2 \right] \label{eq:SDC_TollReverseEnd}
\end{align}

\textbf{Travel time costs} are zero due to zero queue.

\textbf{Toll costs}
\begin{align}
	TRC_1 &= \frac{\beta_2 \eta}{2(1+\eta)}  \frac{N^2}{D} (1-f_2) \left[  2 f_2+ \frac{\beta_1}{\beta_2} (1-f_2) \right] \label{eq:TRC_TollReverseBegin}\\
	TRC_2 &= \frac{\beta_2 \eta}{2(1+\eta)}  \frac{N^2}{D} f_2^2 \\
	TRC &= \frac{\beta_2 \eta}{2(1+\eta)}  \frac{N^2}{D} \left[1 + \left( \frac{\beta_1}{\beta_2} -1 \right) (1-f_2)^2 \right] \label{eq:TRC_TollReverseEnd}
\end{align}

\textbf{Total costs}
\begin{align}
	TC_1 &= \frac{\beta_2 \eta}{(1+\eta)}  \frac{N^2}{D}  (1-f_2) \left[ 1+  \left( \frac{\beta_1}{\beta_2} - 1  \right) (1-f_2) \right] \label{eq:TC_TollOrderDiffBegin}\\
	TC_2 &= \frac{\beta_2 \eta}{(1+\eta)}  \frac{N^2}{D} f_2 \\
	TC &= \frac{\beta_2 \eta}{(1+\eta)}  \frac{N^2}{D}  \left[ 1 + \left( \frac{\beta_1}{\beta_2} - 1 \right) (1-f_2)^2 \right] \label{eq:TC_TollOrderDiffEnd}
\end{align}
As we can compare the total costs, we can see that group 1 is better off under SO toll relative to no toll, while group 2 is worse off.

\subsubsection{$\text{TE1}$ toll}

\textbf{Schedule delay costs} are same as the no-tolled scenario as the order is preserved (Equations~\eqref{eq:SDC_noTollBegin}--\eqref{eq:SDC_noTollEnd}).

\textbf{Travel time costs} are zero due to zero queue.

\textbf{Toll costs} have the same value as $TTC$ of no-toll scenario.
\begin{align}
	TRC_1 &= \frac{\beta_1 \eta}{2(1+\eta)}  \frac{N^2}{D} (1- f_2)^2 \label{eq:TRC_noTollBegin} \\
	TRC_2 &= \frac{\beta_1 \eta}{2(1+\eta)}  \frac{N^2}{D}  f_2 \left[ 2\frac{\alpha_2}{\alpha_1} + \left( \frac{\beta_2}{\beta_1} - 2 \frac{\alpha_2}{\alpha_1} \right) f_2 \right] \\
	TTC &= \frac{\beta_1 \eta}{2(1+\eta)}  \frac{N^2}{D}  \left[ 1- 2 \left( 1- \frac{\alpha_2}{\alpha_1} \right)f_2 + \left( 1+ \frac{\beta_2}{\beta_1} -2 \frac{\alpha_2}{\alpha_1}\right) f_2^2 \right] \label{eq:TRC_noTollEnd}
\end{align}

\textbf{Total costs} are also identical to the no-toll scenario (Equations~\eqref{eq:TC_noTollBegin}--\eqref{eq:TC_noTollEnd}).

\subsubsection{$\text{TE2}$ toll}
\textbf{Schedule delay costs} are same as the no-tolled scenario as the order is preserved (Equations~\eqref{eq:SDC_noTollBegin}--\eqref{eq:SDC_noTollEnd}).

\textbf{Travel time costs} are zero due to zero queue.

\textbf{Toll costs} are in total same as the SO toll since $\text{TE2}$ toll is revenue-neutral but the individual group tolls differ.

\begin{align}
	TRC_1 &= \frac{\beta_1 \eta}{2(1+\eta)}  \frac{N^2}{D} (1-f_2)^2 \label{eq:TRC_TollReverseBeginTE}\\
	TRC_2 &= \frac{\beta_2 \eta}{2(1+\eta)}  \frac{N^2}{D} (2f_2- f_2^2) = \frac{\beta_1 \eta}{2(1+\eta)}  \frac{N^2}{D} \frac{\beta_2}{\beta_1}(2f_2- f_2^2)\\
	TRC &= \frac{\beta_1 \eta}{2(1+\eta)}  \frac{N^2}{D} \left[1+ \left( f_2^2 - 2 f_2\right) \left( 1- \frac{\beta_2}{\beta_1} \right) \right] \label{eq:TRC_TollReverseEndTE}
\end{align}

\textbf{Total costs} compute to the following upon adding the schedule delay costs and toll costs.

\begin{align}
	TC_1 &= \frac{\beta_1 \eta}{(1+\eta)}  \frac{N^2}{D} (1- f_2) \\
	TC_2 &= \frac{\beta_2 \eta}{(1+\eta)}  \frac{N^2}{D} f_2 \\
	TC &= \frac{\beta_1 \eta}{1+\eta}  \frac{N^2}{D} \left[ 1- \left( 1- \frac{\beta_2}{\beta_1} \right) f_2 \right]
\end{align}

Total costs are overall higher for group 2, but remain unchanged for group 1.

\section{Proof of Theorem~\ref{thm:propTravTEtoll}}
\label{appen:proofThm1}
Here we provide the proof that the proportion of travelers departing early is same for both groups. The proof follows the same structure as Appendix B of~\cite{arnott1987schedule}.

\begin{lemma}
The $\zeta$ values at the points of intersection of the equilibrium isocost$-\zeta$ curves are identical. That is, $\zeta(t_{A}^\text{TE})=\zeta(t_{B}^\text{TE})$
\label{lem:equalZeta}
\end{lemma}
\begin{proof}
For the equilibrium isocost$-\zeta$ curves, the costs incurred are identical for all travelers within each group. That is, $C^{\text{eq}}_1(t_{A}^\text{TE})=C^{\text{eq}}_1(t_{B}^\text{TE})$ and $C^{\text{eq}}_2(t_{A}^\text{TE})=C^{\text{eq}}_2(t_{B}^\text{TE})$. Writing out the expressions:

\begin{align}
	C^{\text{eq}}_1(t_{A}^\text{TE}) & =C^{\text{eq}}_1(t_{B}^\text{TE}) \\
	\Longrightarrow \alpha_1 \zeta(t_{A}^\text{TE}) + \beta_1(\tau^*-t_{A}^\text{TE}) &= \alpha_1 \zeta(t_{B}^\text{TE}) + \gamma_1(t_{B}^\text{TE}-\tau^*) \label{eq:appenBcost1}
\end{align}

Dividing Equation~\ref{eq:appenBcost1} by $\beta_1$ and substracting the result from the corresponding equation for group $2$, we obtain:

\begin{equation}
\left(\frac{\alpha_1}{\beta_1} - \frac{\alpha_2}{\beta_2} \right) (\zeta(t_{A}^\text{TE})-\zeta(t_{B}^\text{TE})) = 0
\end{equation}

Given $\alpha_2/\beta_2 < \alpha_1/\beta_1$, we have $\zeta(t_{A}^\text{TE})=\zeta(t_{B}^\text{TE})$. Hence proved.
\end{proof}

Using this Lemma, we can prove Theorem~\ref{thm:propTravTEtoll}.
\begin{proof}
We know that group $1$ arrives in time $(t_0^{\text{TE}},t_{A}^\text{TE}) \cup (t_{B}^\text{TE},t_f^{\text{TE}})$, while group $2$ arrives in time $(t_{A}^\text{TE},t_{B}^\text{TE})$. 

Let us consider group $1$. The increase in the value of $\zeta$ (denoted by $\Delta \zeta$) when group $1$ arrives early is given by:
\begin{equation}
\Delta \zeta = \beta_1 (t_{A}^\text{TE} - t_0^{\text{TE}})
\end{equation}

By Lemma~\ref{lem:equalZeta}, an increase in the value of $\zeta$ is the same while a group $k$ is departing early as the decrease in the value of $\zeta$ when the group $k$ is departing late ($k\in \{1,2\}$). Hence, considering the case when group $1$ arrives late:

\begin{align}
\Delta \zeta &= \gamma_1 (t_f^{\text{TE}}-t_{B}^\text{TE}) \\
\Longrightarrow \beta_1 (t_{A}^\text{TE} - t_0^{\text{TE}}) &= \gamma_1  (t_f^{\text{TE}}-t_{B}^\text{TE}) \\
\Longrightarrow t_{A}^\text{TE} - t_0^{\text{TE}} &= \eta \left( t_f^{\text{TE}}-t_{B}^\text{TE}\right) \label{eq:appenBTimeDiff}
\end{align}

Since the departure rate is equal to the bottleneck capacity, using Equation~\eqref{eq:appenBTimeDiff}, the number of group 1 travelers departing early is $\eta$ times the number of travelers departing late. Thus, $\eta/(1+\eta)$ proportion of group 1 travelers depart early.

We can establish the same result by equating the $\Delta \zeta$ for periods when group 2 travelers arrive early and arrive late. Hence proved.
\end{proof}

We can derive the values of transitions times as follows:

\begin{align}
	t_0^{\text{TE}} &= \tau(t_0^{\text{TE}}) = \tau^* - \frac{\eta}{1+\eta} \frac{N}{D} \\
	t_f^{\text{TE}} &= \tau(t_f^{\text{TE}}) = \tau^* + \frac{1}{1+\eta} \frac{N}{D} \\
	t_{A}^{\text{TE}} &=\tau(t_{A}^{\text{TE}}) = \tau^* - \frac{\eta}{\eta+1}\frac{N_2}{D}  \\
	t_{B}^{\text{TE}} &= \tau(t_{B}^{\text{TE}})= \tau^* + \frac{1}{\eta+1}\frac{N_2}{D}  \label{eq:timeLocationTEToll}
\end{align}

\end{document}